\documentclass[11pt]{siamart190516}

\usepackage{isomath}
\usepackage{amsbsy}
\usepackage{amssymb}
\usepackage{amscd}
\usepackage{amsfonts}
\usepackage{stmaryrd}
\usepackage{siunitx}
\usepackage{euscript}
\usepackage[utf8]{inputenc}
\usepackage[T1]{fontenc}
\usepackage{newtxtext} 
\everymath{\displaystyle}
\usepackage{exscale}

\usepackage{verbatim} 

\usepackage{graphicx}
\usepackage{boxedminipage}
\usepackage{calc}
\usepackage{xcolor}
\graphicspath{ {figure/} }
\usepackage{subcaption}

\usepackage{setspace}
\usepackage{enumitem}
\setitemize{noitemsep,topsep=0pt,parsep=0pt,partopsep=0pt}
\setenumerate{noitemsep,topsep=0pt,parsep=0pt,partopsep=0pt}
\setdescription{noitemsep,topsep=0pt,parsep=0pt,partopsep=0pt}

\newcommand{\R}{\mathbb R}

\newcommand{\calL}{{\cal L}}

\newcommand{\calR}{{\cal R}}
\newcommand{\calS}{{\cal S}}

\newtheorem{remark}{Remark}[section]

\DeclareMathOperator{\divergence}{div}

\newcommand{\dm}{\ \mathrm{d}}
\newcommand{\dmnosp}{\mathrm{d}}

\newcommand{\bfd}{{\mathbold d}}

\newcommand{\bff}{{\mathbold f}}

\newcommand{\bfx}{{\mathbold x}}

\newcommand{\bfQ}{{\mathbold Q}}

\newcommand{\bcolon}{\boldsymbol{:}}
\newcommand{\bbR}{\mathbb{R}}

\newcommand{\bbZ}{\mathbb{Z}}

\newcommand{\dyad}{\mathbf{\otimes}}
\newcommand{\bzero}{\mathbf{0}}
\newcommand{\abs}[1]{\left\lvert #1 \right\rvert}
\newcommand{\norm}[1]{\left\lVert #1 \right\rVert}

\newcommand{\Llambda}{{\cal{L}}_{\lambda}}

\newcommand{\Ltwo}[1]{%
\ifthenelse{\equal{#1}{}}{L^2}{L^2(#1)}%
}

\newcommand{\Ltwoz}[1]{%
\ifthenelse{\equal{#1}{}}{L^2_0}{L^2_0(#1)}%
}

\newcommand{\Cone}[1]{%
\ifthenelse{\equal{#1}{}}{C^{1}}{C^{1}(#1)}%
}

\newcommand{\Conez}[1]{%
\ifthenelse{\equal{#1}{}}{C^{1}_{0}}{C^{1}_{0}(#1)}%
}

\newcommand{\Ctwo}[1]{%
\ifthenelse{\equal{#1}{}}{C^{2}}{C^2(#1)}%
}

\newcommand{\Ctwoz}[1]{%
\ifthenelse{\equal{#1}{}}{C^{2}_{0}}{C^{2}_{0}(#1)}%
}

\newcommand{\Cholder}[1]{%
\ifthenelse{\equal{#1}{}}{C^{0,\gamma}}{C^{0,\gamma}(#1)}%
}

\newcommand{\Cholderz}[1]{%
\ifthenelse{\equal{#1}{}}{C^{0,\gamma}_{0}}{C^{0,\gamma}_{0}(#1)}%
}

\newcommand{\LtwoOneD}{\Ltwo{\bbR,\bbR^3}}
\newcommand{\CtestOneD}{C^{\infty}_{0}(\bbR,\bbR^3)}

\newcommand{\bolds}[1]{\boldsymbol{#1}}
\newcommand{\ba}{\bolds{a}}
\newcommand{\bb}{\bolds{b}}

\newcommand{\bd}{\bolds{d}}
\newcommand{\be}{\bolds{e}}

\newcommand{\bn}{\bolds{n}}

\newcommand{\br}{\bolds{r}}
\newcommand{\bs}{\bolds{s}}
\newcommand{\bt}{\bolds{t}}
\newcommand{\bu}{\bolds{u}}
\newcommand{\bv}{\bolds{v}}

\newcommand{\bx}{\bolds{x}}
\newcommand{\by}{\bolds{y}}
\newcommand{\bz}{\bolds{z}}

\newcommand{\bA}{\bolds{A}}
\newcommand{\bB}{\bolds{B}}

\newcommand{\bH}{\bolds{H}}
\newcommand{\bI}{\bolds{I}}

\newcommand{\bK}{\bolds{K}}

\newcommand{\bP}{\bolds{P}}
\newcommand{\bQ}{\bolds{Q}}

\newcommand{\bdlambda}{\bd_{\lambda}}

\newcommand{\bxbar}{\bar{\bx}}
\newcommand{\thetab}{\bar{\theta}}

\newcommand{\sref}[2]{\hyperref[#2]{#1 \ref*{#2}}}

\newcommand{\Ltwodot}[2]{\left({#1}, {#2}\right)}
\newcommand{\levert}{\left\vert}
\newcommand{\rivert}{\right\vert}

\newcommand{\TheTitle}{Discrete-to-Continuum Limits of Long-Range Electrical Interactions in Nanostructures}

\title{{\TheTitle}\thanks{\today}}

\author{
  Prashant K. Jha\thanks{Oden Institute for Computational Engineering and Sciences, The University of Texas at Austin, Austin, USA (\email{pjha@utexas.edu}).}
  \and
  Timothy Breitzman\thanks{Air Force Research Laboratory, Wright-Patterson Air Force Base, USA (\email{timothy.breitzman.1@us.af.mil}).}  
  \and
  Kaushik Dayal\thanks{Department of Civil and Environmental Engineering, Carnegie Mellon University, Pittsburgh, USA; Center for Nonlinear Analysis, Department of Mathematical Sciences, Carnegie Mellon University, Pittsburgh, USA; Department of Mechanical Engineering, Carnegie Mellon University, Pittsburgh, USA (\email{Kaushik.Dayal@cmu.edu}).}  
}

\begin{document}
	

\maketitle


 \begin{center}
    \textbf{To appear in Archive for Rational Mechanics and Analysis (DOI: \url{https://doi.org/10.1007/s00205-023-01869-6})}
\end{center}

\begin{abstract}
	We consider electrostatic interactions in two classes of nanostructures embedded in a three dimensional space: (1) helical nanotubes, and (2) thin films with uniform bending (i.e., constant mean curvature).
	Starting from the atomic scale with a discrete distribution of dipoles, we obtain the continuum limit of the electrostatic energy; the continuum energy depends on the geometric parameters that define the nanostructure, such as the pitch and twist of the helical nanotubes and the curvature of the thin film. 
	We find that the limiting energy is local in nature. 
	This can be rationalized by noticing that the decay of the dipole kernel is sufficiently fast when the lattice sums run over one and two dimensions, and is also consistent with prior work on dimension reduction of continuum micromagnetic bodies to the thin film limit.
	    However, an interesting contrast between the discrete-to-continuum approach and the continuum dimension reduction approaches is that the limit energy in the latter depends only on the normal component of the dipole field, whereas in the discrete-to-continuum approach, both tangential and normal components of the dipole field contribute to the limit energy.
\end{abstract}
	
\maketitle


\section{Introduction}\label{s:introduction}

Electrical and magnetic interactions are long-range; that is, a charge or dipole interacts with all the other charges and dipoles in the system, and the interactions cannot be truncated because the decay with distance is slow \cite{toupin1956elastic, brown1963micromagnetics, james1994internal, marshall2014atomistic}.
We consider such electrostatic interactions in nanostructures, specifically helical geometries and thin films with uniform bending, in a three-dimensional ambient space.
These geometries are ubiquitous in nanotechnology; while not periodic, their structure has significant symmetry that we exploit in this paper, using the framework of Objective Structures \cite{james2006objective}.
We exploit this symmetry to adapt periodic calculations of the continuum energy to the setting of these nanostructures.
Specifically, starting from a discrete atomic-scale description of the electrostatic energy, we find the limit energy when the discrete lengthscale of the nanostructures goes to zero.

    For simplicity and clarity, we assume in this paper that the charge density can be approximated as composed of discrete dipoles.
    The electrostatic energy of such a system is the sum of all pairwise dipole-dipole interactions.
    Unlike short-range bonded atomic interactions that typically scale as $r^{-6}$ with distance $r$, the dipole-dipole interactions decay slowly with distance as $r^{-3}$.
    Consequently, we cannot simply truncate after a few neighbors, and naive truncation can lead to qualitatively incorrect results in numerical calculations \cite{marshall2014atomistic,grasinger2020architected,grasinger2020statistical}.
    While we use the setting of discrete electrical dipoles, the setting of magnetic dipoles has an identical mathematical structure and physical interpretation \cite{brown1963micromagnetics,james1994internal,muller2002discrete,schlomerkemper2009discrete}, and we borrow key ideas from that literature.
    A key physical distinction between the electrical and magnetic situations is the possibility of electrical monopoles that does not exist for magnetic case, but we examine this elsewhere \cite{shoham-in-prep} and assume here that there are no free charges. 
    Further, we highlight that the assumption of discrete dipoles is not very restrictive.
    Following the approach of \cite{james1994internal}, we use a background field in our calculations, and this field enables a straightforward generalization to the more realistic setting of a general charge density field; such an approach was used by \cite{xiao2005influence} to study charge density fields in periodic crystals.

We turn to the question of dealing with the non-periodic geometry of the nanostructures.
While neither helices nor thin films with curvature are periodic, the framework of Objective Structures (OS) introduced in \cite{james2006objective} provides a powerful approach to deal with such geometries.
In brief, OS provides a group-theoretic description of these nanostructures that enables a parallel to be made with periodic lattices.
This parallel to periodic lattices has enabled the adaptation of various methods developed for lattices to the setting of helices and thin-films, e.g. \cite{dumitricua2007objective,hakobyan2012objective,aghaei2013anomalous,aghaei2013symmetry}.
Our strategy in this work is to use the OS framework to adapt continuum limit calculations from the setting of periodic lattices to the setting of nanostructures.

Our work is focused on obtaining discrete-to-continuum limits of the energy.
This multiscale approach has proven very powerful in enabling the systematic reduction of the very large number of degrees of freedom associated with the discrete problem to a much more tractable continuum problem.
This overall idea has played an important role in developing models, often in conjunction with variational tools such as $\Gamma$-convergence, both for bulk crystals \cite{blanc2002molecular, blanc2007atomistic, schmidt2009derivation, alicandro2008variational, alicandro2011integral, bach2020discrete} as well as for thin films and rods \cite{schmidt2006derivation,alicandro2018derivation}.
Further, these ideas have played a role in the development of numerical multiscale atomistic methods such as the quasicontinuum method \cite{tadmor1996quasicontinuum, miller2002quasicontinuum, tadmor2011modeling,knap2001analysis,li2012positive,dobson2010sharp}. 
    There is also a significant literature on simultaneous dimension reduction and discrete-to-continuum limits, e.g. \cite{alicandro2008continuum,alicandro2018derivation,friesecke2000scheme,lazzaroni2015discrete,lazzaroni2017rigidity,schmidt2008passage}, but these consider interactions that decay much faster than electrostatic interactions.

All the work in the previous paragraph is restricted to the setting of short-range bonded atomic interactions.
In the context of electrical and magnetic interactions, the calculation of continuum limit energies based on discrete-to-continuum approaches have been examined both formally and rigorously using discrete dipoles on a periodic 3-d lattice \cite{toupin1956elastic, brown1963micromagnetics, james1994internal,muller2002discrete, schlomerkemper2009discrete,schlomerkemper2005mathematical}. 
Further, this has been examined formally for periodic charge distributions, also in 3-d, \cite{marshall2014atomistic}. 
All of these works show that the continuum limit energy consists of a local part and a nonlocal part.
In contrast, in this work, we consider topologically low-dimensional structures embedded in a 3-d space: a 1-d helical nanotube and a 2-d thin film with constant bending curvature.
In the limit that the discrete lengthscale characterizing the nanotube and thin film goes to $0$, we find that the limit continuum energy is entirely local. 
    In \cite{cicalese2009discrete}, they obtain the discrete-to-continuum limit of an energy that includes dipole-dipole interactions in 2-d materials using $\Gamma-$convergence, but they use a dipole kernel that goes as $1/r^2$. In contrast, in this work, while the structures are low dimensional, they are embedded in 3-d, and, therefore, the dipole field kernel has a $1/r^3$ singularity.

The absence of nonlocality in the limit can be rationalized by observing that the decay of the interactions as $r^{-3}$ is sufficiently fast to enable us to obtain a local limit if summed over a (topologically) 1-d or 2-d object.
We highlight a complementary body of work that applies dimension reduction techniques to go from a 3-d continuum to a 2-d or 1-d continuum.
In the context of electrical and magnetic interactions, \cite{gioia1997micromagnetics} and subsequent works \cite{carbou2001thin, kohn2005another, kruvzik2015quasistatic} (for thin films) and \cite{gaudiello2015reduced, chacouche2015ferromagnetic} (for thin wires)
find, as we do, that the limit energy is not nonlocal.

The techniques employed in this work are broadly based on the rigorous results provided in \cite{james1994internal} on the continuum limit of magnetic dipole interactions on a 3-d lattice, with appropriate generalizations and modifications for our setting.
The overall strategy of \cite{james1994internal} is as follows.
First, the operator that associates the discrete dipole lattice to the generated electric field is shown to be bounded for smooth test functions; next, the pointwise limit of the action of the operator on smooth test functions is obtained; and, finally, using the boundedness of the operator and the density of the test functions, the limit of the energy density is obtained.
For the helical and thin film nanostructures considered in this work, we adapt this strategy to account for the fact that the lattice sites and dipoles are not related by a  translation transformation, but by a more general isometric transformation.

The key results of this work are as follows.
First, the limit energy is rigorously derived and found to be local.
Second, the limiting energy density depends on the macroscopic geometric parameters, such as the pitch, radius and so on for the helical nanotube, and on the stretch and curvature for the thin film.
These parameters can be related to macroscopic measures of deformation, and link the macroscopic deformation to the small-scale structure. 
Third, while the limiting energy is local, there are energetic contributions from both the normal and the tangential components of the dipole field.
    This is in contrast to the result obtained by dimension reduction from a 3-d continuum:
    in these approaches, there are no energetic contributions from the tangential component of the dipole field.
    Those approaches have 3-d continuum theory as their starting point, and are valid for situations in which the limiting thin object has all dimensions much larger than the atomic lengthscale.
    In contrast, the discrete-to-continuum approach used here is appropriate for nanostructures in which the thin dimensions are comparable to the atomic lengthscale.

\paragraph{Organization}

In \autoref{s:motivation}, we discuss prior work, primarily on dimension reduction from a 3-d continuum to a 2-d continuum, and highlight the local nature of the limiting energy.  We then discuss heuristically the scaling of electrostatic interactions that lead to this locality generically for topologically low-dimensional nanostructures.
In \autoref{s:limit}, we present the main results for helical nanotubes and thin films with constant bending curvature. 
We prove various claims in \autoref{s:proof}. 
In \autoref{s:conclusion}, we summarize the results.

\paragraph*{Notation}

We denote the real line and set of integers by $\bbR$ and $\bbZ$, respectively; $\bbR^d, \bbZ^d$ denote these in dimension $d = 1,2,3$. 
For any $c, c_1, c_2 \in \bbR$, $c\bbZ^d$ denotes the set $\{cz ; z\in \bbZ^d \}$ and $c_1\bbZ \times c_2 \bbZ$ denotes the set $\{(c_1 z_1, c_2 z_2); z_1, z_2 \in \bbZ\}$.
The symbols $\calL$ and $U$ denote the set of lattice sites and the lattice unit cell, respectively; $\calL_\lambda$ and $U_\lambda$ denote these in the lattice scaled by $\lambda$, with $\calL_1, U_1$ denoting $\calL_\lambda, U_\lambda$ for $\lambda = 1$.	%
We use $\bx = (x_1, x_2, x_3) \in \bbR^3$ to denote the point in space with components $x_i$ in the orthonormal basis $\{\be_1, \be_2, \be_3\}$ for $\bbR^3$. 
We follow the standard notation wherein scalars are denoted by lowercase letters, vectors by bold lowercase letters, and second order tensors by bold uppercase letters. 
$|\bx| = \sqrt{ \sum_{i=1}^n x_i^2 }$ denotes the Euclidean norm of the vector $\bx\in \bbR^n$;  $|\bA| = \sqrt{\bA\bcolon \bA}$ denotes the norm of the tensor $\bA$; and $\bA\bcolon \bB = A_{ij} B_{ij}$ denotes the inner product of the tensors $\bA$ and $\bB$. For any vector $\ba\in \bbR^n$ and tensor $\bA$, we have $|\bA \ba| \leq |\bA| |\ba|$.
We use $|\varOmega|$ to denote the Lebesgue measure of the set $\varOmega \subset \bbR^n$. 
For a set $A \subset \bbR^d$, $\chi_A = \chi_A(\bx)$ denotes the indicator function. 
We use $L^2(A, B)$ to denote the space of Lebesgue square-integrable functions $u\colon A \subset\bbR^n \to B\subset \bbR^m$; $(u, v)_{L^2(A,B)}$ for the inner product of functions $u, v \in L^2(A,B)$; and $||u||_{L^2(A,B)}$ the $L^2$ norm of $u\in L^2(A,B)$. 
When there is no ambiguity, we will suppress $L^2(A,B)$ and write $(u,v)$ and $||u||$. 
$C^\infty_0(\bbR^n, \bbR^m)$ denotes the space of infinitely differentiable test functions $u\colon \bbR^n \to \bbR^m$ with compact support in $\bbR^n$. 
For Hilbert spaces $V$ and $W$, $\calL(V, W)$ is the space of bounded linear maps $T\colon V \to W$. The norm of the map $T\in \calL(V,W)$ is denoted by $||T||_{\calL(V,W)}$, and is given by the expression:
\begin{equation}\label{eq:normMap}
	||T||_{\calL(V,W)}  = \sup_{||f||_V \neq 0} \frac{||T f||_{W}}{||f||_{V}}.
\end{equation}
We use $u_\lambda \xrightarrow[\lambda \to 0]{} u$ to denote the strong convergence of $u_\lambda \in V$ to $u\in V$ as $\lambda \to 0$, i.e., $||u_\lambda - u||_V \to 0$ as $\lambda \to 0$.

\section{Energy Scalings and Prior Results on Dimension Reduction}
\label{s:motivation}

We briefly revisit the results of \cite{gioia1997micromagnetics} and \cite{chacouche2015ferromagnetic}.
Respectively, they performed dimension reduction from the 3-d continuum to the 2-d thin film and 1-d thin wire to find the limiting magnetostatic energy.

\begin{figure}[htb!]
	\centering
	\includegraphics[scale=0.2]{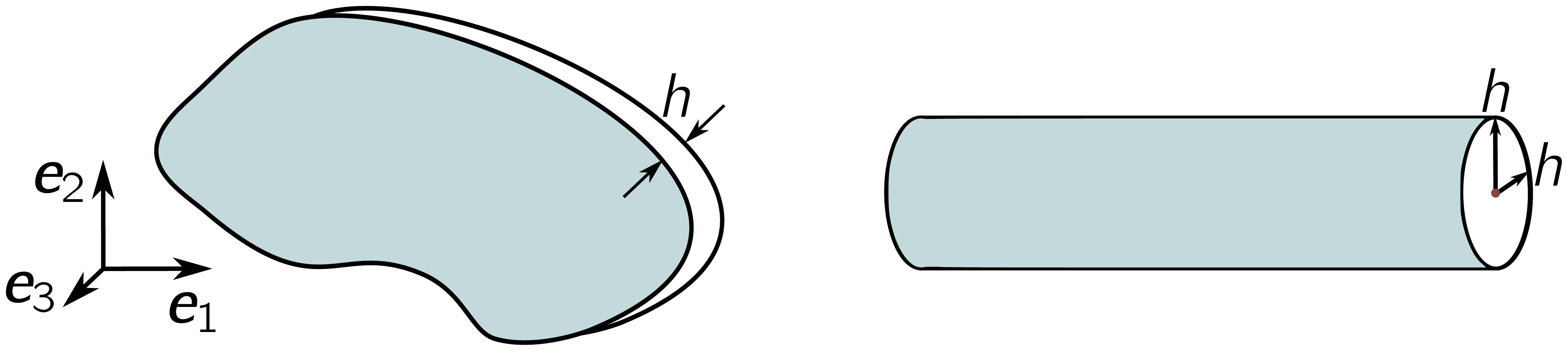}
	\caption{The geometry of the thin film (left) and the thin wire (right).}
	\label{fig:thinstruct}
\end{figure}

Consider a material domain $\varOmega_h = S \times [0,h]$, where $S \subset \bbR^2$ is a 2-d domain in the plane spanned by $(\be_1, \be_2)$, and $h>0$ is the material thickness in the normal direction $\be_3$ (Fig. \ref{fig:thinstruct}).
Suppose $\bd\colon \varOmega_h \to \bbR^3$, with $\bd = \bzero$ on $\bbR^3 \backslash \varOmega_h$, is the dipole field in the material. 
The electrostatic energy density is given by
\begin{equation*}
	e_h(\bd) = \dfrac{1}{|\varOmega_h|} \int_{\varOmega_h} \frac{1}{2} \nabla \phi(\bx) \cdot \bd \dm\bx,
\end{equation*}
where $|\varOmega_h|$ is the volume of $\varOmega_h$, and $\phi$ is the electric potential that satisfies the electrostatic equation 
\begin{equation*}
	\divergence (-\nabla \phi + \bd) = 0 \quad  \text{on } \bbR^3
\end{equation*}
together with the constraint $|\bd| = d$ and the decay property
\begin{equation*}
	 |\nabla\phi(\bfx)| \to {\bf 0} \text{ as } |\bfx|\to\infty .
\end{equation*}
Let $\varOmega_1 = S \times [0,1]$, and $\by(\bx) = (x_1,x_2,x_3/h) \in \varOmega_1$ for $\bx \in \varOmega_h$ be the map from $\varOmega_h$ to $\varOmega_1$. 
For fixed $h>0$, consider the dipole field $\bd_h\colon \varOmega_h \to \bbR^3$ and $\tilde{\bd}_h\colon \varOmega_1 \to \bbR^3$ such that
\begin{equation*}
	\tilde{\bd}_h(\by(\bx))  = \bd_h(\bx), \qquad \forall \bx \in \varOmega_h.
\end{equation*}
Let $\bd_h$ be the sequence of dipole field for $h>0$, and $\tilde{\bd}_h$ is defined as above. 
Assume that dipole field $\tilde{\bd}_h$ is such that, first, $\tilde{\bd}_h = 0$ on $\bbR^3 \backslash \varOmega_1$, and, second, it converges to $\tilde{\bd}_0$ in $L^2(\bbR^3)$; then the limit of the energy density $e_h = e_h(\bd_h)$ is \cite{gioia1997micromagnetics}:
\begin{equation*}
	e_h(\bd_h) \to e_0(\tilde{\bd}_0) 
	= 
	\dfrac{1}{2|\varOmega_1|} \int_{\varOmega_1} |\tilde{d}_{0_3}|^2 \dm\bx.
\end{equation*}
That is, the limiting energy $e_0$ is local, and only the normal component of the dipole moment appears in the expression.

Next, consider a thin straight wire with axis along $\be_1$, denoted by $\varOmega_h = (-1, 1) \times B_2(\bzero, h)$, where $B_2(\bzero, h)$ is the ball of radius $h$ centered at $\bzero$ in the plane spanned by $(\be_2, \be_3)$ (Fig. \ref{fig:thinstruct}). 
    Analogous to the thin film, let 
    $\bd_h: \varOmega_h \to \bbR^3$, $\bd_h = \bzero$ on $\bbR^3/\varOmega_h$, be the dipole field in the material;
    $\varOmega_1 = (-1, 1) \times B_2(\bzero, 1)$ be the rescaled domain of $\varOmega_h$; 
    $\by(\bx) = (x_1, x_2/h, x_3/h) \in \varOmega_1$ for $\bx \in \varOmega_h$ be the map from $\varOmega_h$ to $\varOmega_1$; 
    and $\tilde{\bd}_h : \varOmega_1 \to \bbR^3$ be the rescaled dipole field defined as
    \begin{equation}
        \tilde{\bd}_h(\by(\bx)) = \bd_h(\bx), \qquad  \forall \bx \in \varOmega_h .
    \end{equation}
The limiting energy density in the case of the thin wire is \cite{chacouche2015ferromagnetic}:
\begin{equation*}
	 \frac{1}{2|\varOmega_1|}\int_{\varOmega_1} \left(|\tilde{d}_{0_2}|^2 + |\tilde{d}_{0_3}|^2\right) \dm\bx, 
\end{equation*}
where $\tilde{\bd}_0 := \lim_{h\to 0} \tilde{\bd}_h$ is the limiting dipole field.
We notice that the limiting energy is again local, and only the components of the dipole moment perpendicular to the axis of the wire contribute.

\begin{figure}[htb!]
	\centering
	\includegraphics[scale=0.2]{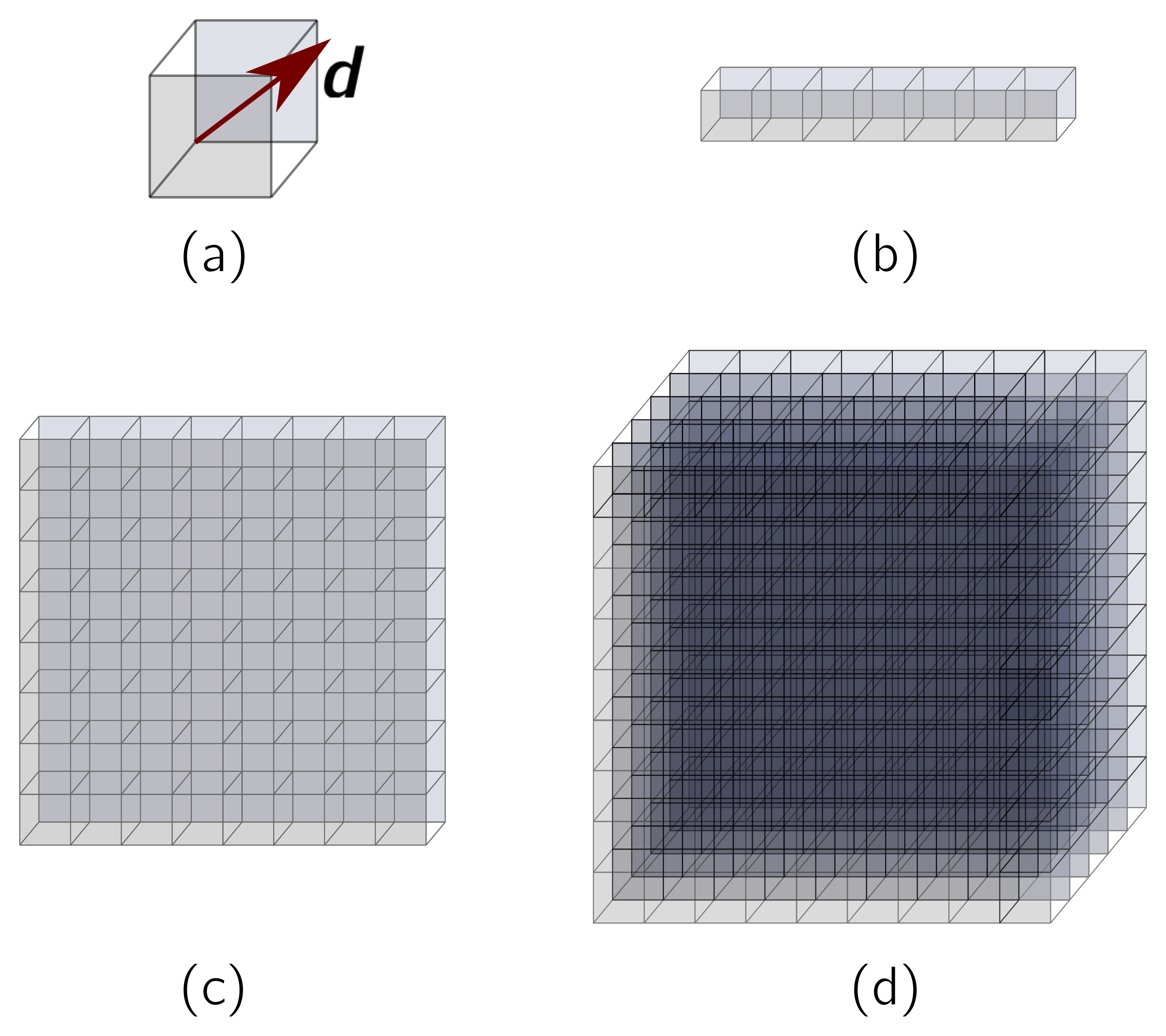}
	\caption{A schematic of the unit cell with a dipole (a), and generic 1-d, 2-d, 3-d periodic lattices (b, c, d).}\label{fig:lattices}
\end{figure}

The absence of nonlocality in the limiting energy in the results above, as well as in our results in section \ref{s:limit} below, can be physically understood through the fact that these structures are 1-d or 2-d topologically.
To see this, we consider a system of discrete dipoles associated with the uniform 1-d, 2-d, and 3-d periodic lattices with the unit cell of size $1$ (\autoref{fig:lattices}). 
The energy of a lattice of dipoles is given by \cite{brown1963micromagnetics, james1994internal}:
\begin{equation}\label{eq:discElectricalEnergy}
	E = -\frac{1}{2}\sum_{i} \sum_{j, j\neq i} \bd(i) \cdot \bK(\bx_j - \bx_i) \bd(j) = \sum_{i} |U(i)| \left[-\frac{1}{|U(i)|} \frac{1}{2} \sum_{j, j \neq i} \bd(i) \cdot \bK(\bx_j - \bx_i) \bd(j)\right], 
\end{equation}
where the sum is over the cells in the lattice, and the term inside square bracket denotes the energy density of a cell $i$. Here, $U(i)$ denotes the $i^{\text{th}}$ unit cell; $|U(i)|$ the measure (volume) of the unit cell $U(i)$; $\bd(i)$ the dipole in cell $i$; and, $\bx_i$ the coordinate of lattice site $i$. The dipole field kernel, $\bK = \bK(\bx)$, is defined as:
\begin{equation}\label{eq:dipoleKernel}
	\bK(\bx) = - \dfrac{1}{4 \pi |\bx|^3} \left( \bI - 3\dfrac{\bx}{|\bx|} \dyad \dfrac{\bx}{|\bx|}\right), \quad \bx \neq \bzero .
\end{equation}
We use these expressions to heuristically understand the scaling of the energy for systems with different topological dimensions.
For simplicity, we assume below that the volume of the unit cell and the magnitude of the dipole are both $1$, i.e., $|U(i)| = 1$ and $|\bd(i)| = 1$ for each $i$, and some constant factors are neglected. For the next set of bounds, assume that $c_1$ and $c_2$ are generic positive constants.

\begin{remark}[1-d lattice]
    We can estimate an upper bound on the energy density $e$ of a typical unit cell as follows:
    \begin{equation*}
    	|e| \leq c_1 \sum_{r = 1}^\infty \dfrac{1}{r^3} \times |\bd| \times (\text{number of dipoles at }r )
    	\leq c_2 \sum_{r = 1}^\infty \dfrac{1}{r^3} \times 1 \times 1 = c_2 \sum_{r=1}^{\infty} \dfrac{1}{r^3} .
    \end{equation*}
    We use that the total dipole moment at a distance $r$ from a given unit cell is, at most, that of another dipole in the unit cell at a distance $r$.
    This sum is well-behaved and bounded. 
\end{remark}

\begin{remark}[2-d lattice]
    As in the 1-d setting, we first bound the net dipole at a distance $r$ from a given unit cell.
    Since the structure is a 2-d lattice, the number of unit cells at a distance $r$ is of order $2\pi r$.
    Therefore, an upper bound on the energy density is:
    \begin{equation*}
    	|e| \leq  c_1 \sum_{r = 1}^\infty \dfrac{1}{r^3} \times |\bd| \times (\text{number of dipoles at }r)
    	\leq c_2 \sum_{r = 1}^\infty \dfrac{1}{r^3} \times 1 \times 2 \pi r = c_2 2\pi \sum_{r=1}^{\infty} \dfrac{1}{r^2}.
    \end{equation*}
    This sum is also well-behaved and bounded. 
\end{remark}

\begin{remark}[3-d lattice]
    Following the argument of the 2-d lattice, we now have that the net dipole at a distance $r$ from a given unit cell is, at most, of the order $4\pi r^2$.
    Therefore, an upper bound on the energy density is:
    \begin{equation*}
    	|e| \leq c_1 \sum_{r = 1}^\infty \dfrac{1}{r^3} \times |\bd| \times (\text{number of dipoles at }r) 
    	\leq c_2 \sum_{r = 1}^\infty \dfrac{1}{r^3} \times 1 \times 4\pi r^2 = c_2 4\pi \sum_{r=1}^{\infty} \dfrac{1}{r} .
    \end{equation*}
    This sum is divergent.
    However, through a more careful analysis that accounts for the signs -- not just the magnitudes -- of the dipole interactions, the energy density can be shown to be conditionally convergent\footnote{
        That is, it is convergent, while the sums of only the positive terms and only the negative terms diverge, respectively, to $\pm\infty$.
    } \cite{toupin1956elastic, james1994internal}.
\end{remark}

When the lattice sum is bounded and converges unconditionally, it is possible to truncate after a finite distance and obtain sufficient numerical accuracy.
When the lattice sum is conditionally convergent, that can be physically related to nonlocality; specifically, the slow convergence does not allow for truncation, and the far-field values play an important role.
\cite{marshall2014atomistic} discusses this from a physical perspective.

 	\section{Results on Continuum Limits of the Electrostatic Energy}
	\label{s:limit}
	
	We consider two classes of nanostructures: helical nanotubes and thin films, the latter allowing for a constant bending curvature (i.e., nonzero constant mean curvature and zero Gauss curvature), and obtain the corresponding continuum limit electrostatic energy.
    In both cases, we start with discrete dipoles, where the discreteness is parametrized by the scale $\lambda > 0$, and examine the limit $\lambda \to 0$. 
	We show that the dipole-dipole interaction energy density -- per unit cross-sectional area in the case of nanotubes, and per unit thickness in the case of films -- converges to a local energy density in the limit.

 	\subsection{Helical Nanotube}
	\label{ss:helix}
	
	We consider a discrete helix with axis $\be_3$ characterized by the angle $\theta$ and lengthscale $\delta$; the pitch of the helix is $2 \pi \delta / \theta$. Suppose $\bx_0 \in \bbR^3$ is a point on the helix. Then, the other points on the helix are related by an isometric transformation of $\bx_0$. Let $s\in \bbR$ be the parametric coordinate of a point on the helix. Then, the map $\bxbar\colon \bbR \to \bbR^3$ that takes a point in the parametric space to a unique point on the helix can be expressed as
	\begin{equation}\label{eq:helixMap}
		\bxbar(s) = \bQ(s\theta) \bx_0 + s \delta \be_3.
	\end{equation}
	Here $\bQ = \bQ(\alpha)$ is the rotational tensor represented by the matrix in the orthonormal basis $\{\be_1, \be_2, \be_3\}$ as:
	\begin{equation}
        \begin{split}\label{eq:rotTensorQ}
    		\bQ(\alpha) &:= \left[\begin{array}{ccc}
    			\cos(\alpha) & - \sin (\alpha) & 0 \\
    			\sin(\alpha) & \cos (\alpha) & 0\\
    			0 &0 & 1
    		\end{array} \right].
	    \end{split}
    \end{equation}
	Note that the definition of map $\bxbar$ imply that the helix makes a full turn in $s=2\pi / \theta$ from $s=0$, and, therefore, the pitch of the helix is $(\bxbar(2\pi/\theta) - \bxbar(0))\cdot \be_3 = 2\pi\delta/\theta$.
	
	Without loss of generality, we assume $\bx_0 = \be_1$. The tangent vector to the helix at $s$ is given by
	\begin{equation}\label{eq:helixTangent}
		\bt(s) = \frac{\dmnosp \bxbar(s)}{\dmnosp s} = \theta \bQ'(s\theta) \be_1 + \delta \be_3 .
	\end{equation}
	Let $\hat{\bt}(s) = \bt(s) / \sqrt{\theta^2 + \delta^2}$ denote the unit tangent vector. 
	We define the second order projection tensors $\bP_{||} = \bP_{||}(s)$ and $\bP_{\perp} = \bP_{\perp}(s)$, for $s\in \bbR$, as follows
	\begin{equation}\label{eq:helixProj}
	\bP_{||}(s) = \hat{\bt}(s) \dyad \hat{\bt}(s), \qquad \bP_{\perp}(s) = \bI - \bP_{||}(s) .
	\end{equation}
	For any vector $\ba$ and any $s\in \bbR$, we have
	\begin{equation}
	\ba = \bP_{||}(s) \ba + \bP_{\perp}(s) \ba, \qquad \text{with } \qquad  \bP_{||}(s) \ba \cdot \bP_{\perp}(s) \ba = 0.
	\end{equation}
	
	\subsubsection{Lattice Geometry and Dipole Moment}\label{sss:helixDipoleSystem}
	
	Let $\calL = \bbZ$ denote the set of parametric coordinates of the points on the helix. We consider a discrete system of dipole moments $\bd \colon \calL \to \bbR^3$ associated to the points on the helix given by $\calL$ (\autoref{fig:helixLatticeExample}).
	The magnitudes of the dipoles at the lattice sites are equal, but they are oriented differently; in particular, the orientations of dipoles at lattice sites follow the relation:
	\begin{equation}
	\label{eqn:dipole-relation-nanotube}
	    \bfd(s+1) = \bfQ(\theta) \bfd(s), \quad s \in \calL.
	\end{equation}
	We associate a unit cell to each lattice site.
	Let $U(s) = [s, s+1)$ denote the unit cell in the parametric space at the site $s$, for $s\in \calL$.
	Let $S(r)$, $r\in \bbR$, be given by $$S(r) = \left\{\bx; \, (\bx - \bxbar(r)) \cdot \bt(r) = 0, \, |\bx - \bxbar(r)|^2 < R^2 \right\},$$ for some $R > 0$. Note that $|S(r)| = |S(0)| = \pi R^2$. 
	The unit cell in real space is defined by $\bar{U}(s) = \left\{\bfx \in S(r);\, r \in U(s) \right\}$. We take, without loss of generality, $R^2 = 1/(\pi \sqrt{\theta^2 + \delta^2})$ so that $|\bar{U}(s)| = \mathrm{area}(S) \times \mathrm{length}(\{\bxbar(r); r\in U(s)\}) = \pi R^2 \sqrt{\theta^2 + \delta^2} = 1$. 
	
    We now consider the setting in which the cells are of size $\lambda >0$, so that as $\lambda \to 0$ the density of cells in the helix increases. For $\lambda >0$, suppose $\calL_\lambda = \lambda \bbZ$ denotes the parametric coordinates of the sites in a scaled lattice, and $\bdlambda \colon \calL_\lambda \to \bbR^3$ denotes the corresponding system of dipole moments. Associated to $s\in \calL_\lambda$, let $U_\lambda(s) = [s, s+\lambda)$ denote the cell in the parametric space. The 3-d cell is given by $\bar{U}_\lambda(s) = \{\bx \in S_{\lambda}(r);\, r \in U_\lambda(s)\}$, where $S_\lambda(r) = \{\bx; \, (\bx - \bxbar(r))\cdot \bt(r) = 0, \, |\bx - \bxbar(r)|^2 < \lambda^2 R^2\}$ is the scaled cross-section. Note that $\mathrm{area}(S_\lambda(r)) = \pi\lambda^2R^2$ and $|\bar{U}_\lambda(s)| = \pi \lambda^2R^2 \times \lambda \sqrt{\theta^2 + \delta^2} = \lambda^3$. Let $\tilde{\bd}_\lambda \colon \bbR \to \bbR^3$ be the piecewise constant extension of $\bdlambda$ given by
    \begin{equation}\label{eq:helixL2Ext}
        \tilde{\bd}_\lambda(s) = \frac{\bdlambda(i)}{|\bar{U}_\lambda(s)|} = \frac{\bdlambda(i)}{\lambda^3}, \qquad \forall s\in U_\lambda(i), \quad \forall i \in \calL_\lambda .
    \end{equation}
    To compute the limit of the dipole-dipole interaction energy as $\lambda \to 0$, we assume that dipole moment density field $\tilde{\bd}_\lambda$ converges to some field $\bff \in L^2(\bbR, \bbR^3)$ in the $L^2$ norm. As in \cite{james1994internal}, instead of working with $\tilde{\bd}_\lambda$, as defined above, we could assume that the dipole moment $\bdlambda(i)$, for $i \in \calL_\lambda$, is due to the background dipole moment density field $\bff_\lambda \in L^2(\bbR, \bbR^3)$ such that 
    \begin{equation}\label{eq:helixDiscreteDipoleFromBg}
        \bdlambda(i) = \sqrt{\theta^2 + \delta^2} \int_{U_\lambda(i)} \int_{S_\lambda(r)} \bff_\lambda(r) \dm S_\lambda(r) \dm r = \lambda^2 \int_{U_\lambda(i)} \bff_\lambda(r) \dm r ,
    \end{equation}
    where $\dmnosp S_\lambda(r)$ is the area measure for surface $S_\lambda$. In the equation above, we assumed that the background field is uniform in $S_\lambda(r)$ for all $r\in \bbR$ and used $R^2  = 1/(\pi \sqrt{\theta^2 + \delta^2})$. 
    The existence of one such background field $\bff_\lambda$ is evident: we can define $\bff_\lambda = \tilde{\bd}_\lambda$. 
    The physical dimension of $\bff_\lambda$ is dipole moment per unit volume. We have the following lemma that relates the convergence of the background dipole moment field and the piecewise constant extension. 
    
    \begin{lemma}\label{lem:helixBgConv}
        Let $\bff_\lambda$, $\lambda > 0$, be the sequence of $L^2(\bbR, \bbR^3)$ functions, and let $\bff\in L^2(\bbR, \bbR^3)$ be such that $\bff_\lambda \to \bff$ in $L^2(\bbR, \bbR^3)$. Let $\bd_\lambda \colon \calL_\lambda \to \bbR^3$ be given by \eqref{eq:helixDiscreteDipoleFromBg}, and $\tilde{\bd}_\lambda$ be a piecewise constant $L^2$ extension of $\bdlambda$ given by \eqref{eq:helixL2Ext}. Then, $\tilde{\bd}_\lambda \to \bff$ in $L^2(\bbR, \bbR^3)$. 
        
        On the other hand, if $\bdlambda\colon \calL_\lambda \to \bbR^3$ is such that $\tilde{\bd}_\lambda \to \bff$ in $L^2(\bbR, \bbR^3)$, then there exists a background field $\bff_\lambda \in L^2(\bbR, \bbR^3)$ such that $\bdlambda$ is given by \eqref{eq:helixDiscreteDipoleFromBg}. 
    \end{lemma}
    The proof is similar to the proof of Theorem 4.1 from \cite{james1994internal}.

    \begin{remark}
        Since the discrete dipole field $\bdlambda$ has helical symmetry, from \eqref{eq:helixDiscreteDipoleFromBg} we can see that $\bff_\lambda$ will also have helical symmetry.
        However, we highlight that the background field needs to have helical symmetry only in the sense that the effective dipole has the helical symmetry, allowing for some fluctuations from site to site. 
    \end{remark}
    
    \begin{remark}
        The physical setting that we wish to examine is when dipole moments play a significant role in the limit.
        Therefore, we consider an appropriate scaling that corresponds to obtaining a finite limit for the dipole density field $\tilde{\bd}_\lambda$ (or, equivalently, the background field, $\bff_\lambda$).
        Further, we assume convergence in $L^2$ to ensure finite energies.
    \end{remark}

	\begin{figure}[!ht]
	\centering
    \begin{subfigure}[t]{.45\textwidth}
      \centering
      \includegraphics[width=.6\linewidth]{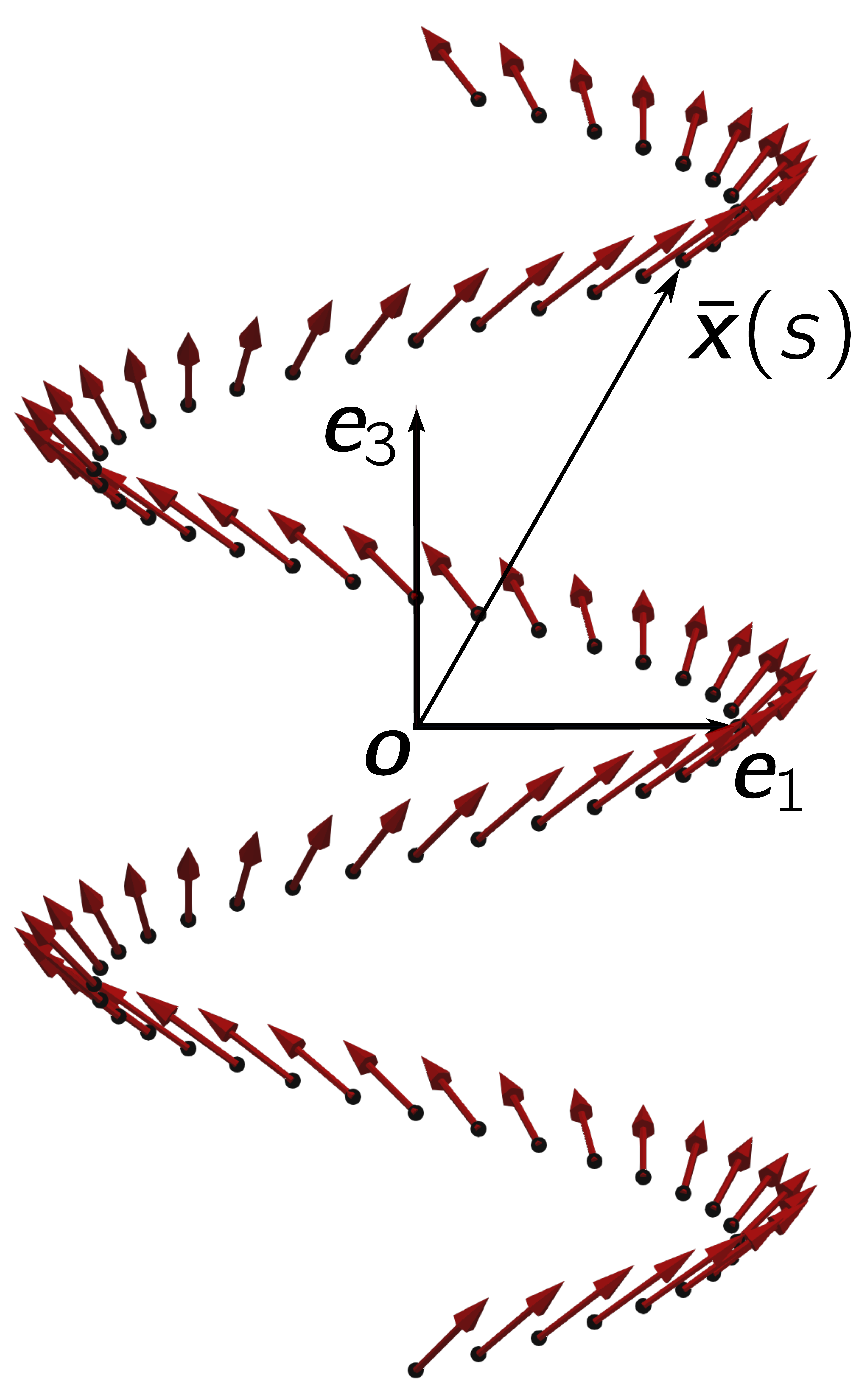}  
      \caption{}
    \end{subfigure}
    \begin{subfigure}[t]{.45\textwidth}
      \centering
      \includegraphics[width=.6\linewidth]{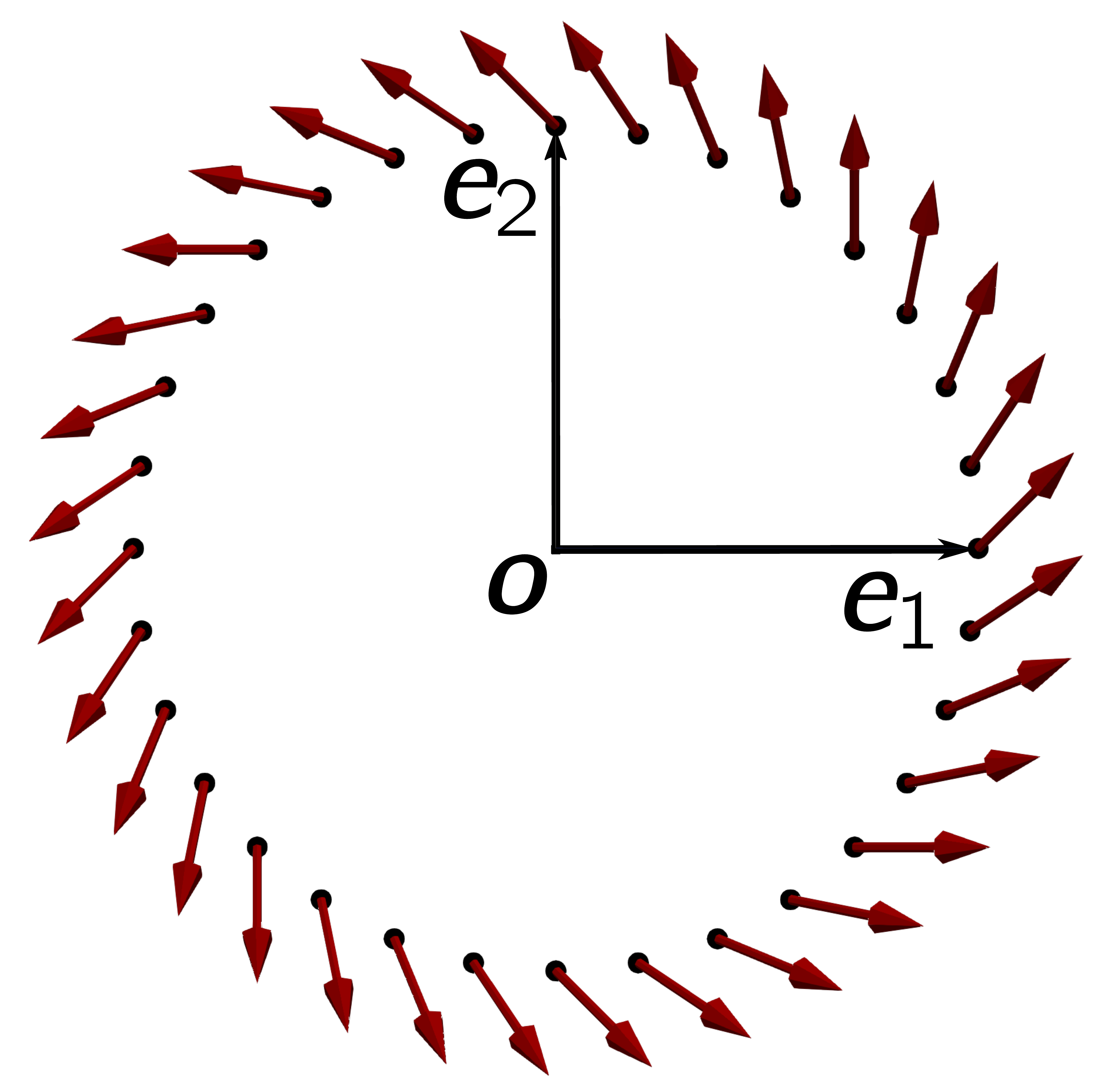}  
      \caption{}
    \end{subfigure}
    \caption{Discrete dipole moments (red arrows) lying on the helix. (a) and (b) show the view in the $(\be_1, \be_3)$ and $(\be_1, \be_2)$ planes respectively.  The dipole moments corresponding to different sites are related by \eqref{eqn:dipole-relation-nanotube}. For the parametric coordinate $s$, $\bxbar(s)$ gives the coordinate of the point on the helix.}
    \label{fig:helixLatticeExample}
    \end{figure}

	\subsubsection{Electrostatic Energy}\label{sss:energyHelix}

	For $\lambda >0$, the energy associated to the system of dipole moments $\bdlambda$ can be expressed as \cite{brown1963micromagnetics, james1994internal}:
	\begin{equation*}
		E_\lambda = -\dfrac{1}{2} \sum_{\substack{s,s'\in \Llambda,\\
				s\neq s'}} \bdlambda(s)\cdot \bK(\bxbar(s') - \bxbar(s)) \bdlambda(s')  = |S_\lambda|e_\lambda,
	\end{equation*}
	where $e_\lambda$ is the energy per unit area given by
	\begin{equation}\label{eq:helixEnergy}
		e_\lambda = -\dfrac{1}{2|S_\lambda|} \sum_{\substack{s,s'\in \Llambda,\\
				s\neq s'}} \bdlambda(s)\cdot \bK(\bxbar(s') - \bxbar(s)) \bdlambda(s').
	\end{equation}
	Substituting \eqref{eq:helixDiscreteDipoleFromBg} into the expression above, and proceeding similar to Section 6 of \cite{james1994internal}, $e_\lambda$ can be written as
	\begin{equation}\label{eq:helixEnergyTLambda}
		e_\lambda = \Ltwodot{\bff_\lambda}{T_\lambda \bff_\lambda}_{L^2(\bbR, \bbR^3)},
	\end{equation}
	where $T_\lambda\colon L^2(\bbR, \bbR^3) \to L^2(\bbR, \bbR^3)$ is the map given by
	\begin{equation}\label{eq:helixTLambda}
		(T_\lambda \bff) (s) = \lambda^2 \int_{\bbR} \bK_\lambda (s',s) \bff(s') \dm s'
	\end{equation}
	and $\bK_{\lambda}(s',s)$, for $s,s'\in \bbR$, is the discrete dipole field kernel given by
	\begin{equation}\label{eq:helixKLambda}
		\bK_{\lambda}(s',s) = \sum_{\substack{u,v\in \Llambda,\\
				u \neq v}} \chi_{U_\lambda(v)}(s') \bK(\bxbar(v) - \bxbar(u)) \chi_{U_\lambda(u)}(s) .
	\end{equation}
	
	\paragraph{Scaling of  $\bK_\lambda$} For any $a,b\in \bbR$, we have, using \eqref{eq:helixMap},
	\begin{equation}
	\label{eq:helixALambda}
    	\begin{split}
    		\bxbar(\lambda a) - \bxbar(\lambda b) &= \bQ(\lambda a\theta) \be_1 + \delta \lambda a \be_3 - \bQ(\lambda b\theta) \be_1 - \delta \lambda b \be_3\\
    		&= \lambda \left( \bxbar(a) - \bxbar(b) 
    		+ 
    		\underbrace{\left[ \dfrac{\bQ(\lambda a\theta) - \bQ(\lambda b\theta) - \left( \lambda \bQ(a\theta) - \lambda \bQ(b\theta) \right)}{\lambda} \right]}_{=:\bA_\lambda(a,b)}
    		\be_1 \right) . 
    	\end{split}
	\end{equation}
	Using the relation above, it is easy to show that
	\begin{equation*}
		\bK_{\lambda}(s',s) = \dfrac{1}{\lambda^3} \sum_{\substack{u,v\in \calL_1,\\
				u \neq v}} \chi_{U_1(v)}(s' /\lambda) \bK(\bxbar(v) - \bxbar(u) + \bA_{\lambda}(v,u) \be_1 ) \chi_{U_1(u)}(s/\lambda) ,
	\end{equation*}
	where we recall that $U_1(u) = [u, u+1)$, $u \in \calL_1$, is the lattice cell in the parametric space for $\lambda = 1$. Based on the equation above, we define a discrete kernel $\bK_{1,\lambda}(s,s')$, for $s,s'\in \bbR$, as follows
	\begin{equation}\label{eq:helixK1Lambda}
		\bK_{1,\lambda}(s',s) = \sum_{\substack{u,v\in \calL_1,\\
				u \neq v}} \chi_{U_1(v)}(s') \bK(\bxbar(v) - \bxbar(u) + \bA_{\lambda}(v,u) \be_1 ) \chi_{U_1(u)}(s) .
	\end{equation}
	We then have
	\begin{equation*}
		\bK_{\lambda}(s',s) = \dfrac{1}{\lambda^3} \bK_{1,\lambda}(s'/\lambda, s/\lambda) .
	\end{equation*}
	
	\subsubsection{Limit of the Electrostatic Energy}
	In this section, we obtain the limit of the energy per unit surface area $e_\lambda$ as $\lambda \to 0$ assuming that the background dipole field density  $\bff_\lambda$ (or equivalently the dipole moment density $\tilde{\bd}_\lambda$) converges to some density field $\bff$ in $L^2$. The idea is to first show that the map $T_\lambda$ in \eqref{eq:helixTLambda} is bounded and obtain its limit. With that, the limit of $e_\lambda$ follows.
	
	\paragraph{Limit of the Discrete Electric Field}
	Let $T_{1,\lambda}$ be the map with kernel $\bK_{1,\lambda}$. For any function $\bff \in L^2(\bbR, \bbR^3)$, we have
	\begin{equation}\label{eq:helixT1Lambda}
		(T_{1,\lambda}\bff)(s) = \int_{\bbR} \bK_{1,\lambda}(s',s) \bff(s') \dm s' .
	\end{equation}
	We have the following main result on the map $T_\lambda$.
	
	\begin{proposition}
		\label{prop:helixTLambda}
		The maps $T_{1,\lambda}$ and $T_{\lambda}$ are bounded in $L^2(\bbR, \bbR^3)$ for all $\lambda > 0$ and satisfy 
		\begin{equation}\label{eq:helixTLambdaEquiv}
			\norm{T_\lambda}_{\mathcal{L}(L^2,L^2)} = \norm{T_{1,\lambda}}_{\mathcal{L}(L^2,L^2)} .
		\end{equation}
		Further, for $\bff\in \CtestOneD$,
		\begin{equation*}
			(T_\lambda \bff)(s) \xrightarrow[\lambda \to 0]{} -h_0 (\bI - 3\bP_{||}(s))\bff(s) = -h_0 (\bP_{\perp}(s) - 2\bP_{||}(s))\bff(s)
 		\end{equation*}
		pointwise, where $\bP_{\perp}(s)$ and $\bP_{||}(s)$ are projection tensors that project onto the normal plane and the tangent line to the helix respectively (see \eqref{eq:helixProj}). Further, $h_0$ is a constant given by
		\begin{equation}\label{eq:helixH0Define}
			h_{0} = \sum_{\substack{v\in \bbZ,\\
						v \neq 0}} \frac{1}{4 \pi |v|^3 (\theta^2 + \delta^2)^{3/2}}.
		\end{equation}
	\end{proposition}
	
	We provide the proof of \sref{Proposition}{prop:helixTLambda} in \autoref{ss:helixProof}. 
		
	\paragraph{Limit of the Energy}
	\begin{theorem}
		\label{thm:helixELambda}
		Let $\bff_\lambda \in \LtwoOneD$ be a sequence of functions for $\lambda >0$ with $\bff \in \LtwoOneD$ such that $\bff_\lambda \xrightarrow[\lambda\to 0]{} \bff$ in $L^2$. Let the system of dipole moments $\bdlambda\colon \calL_\lambda \to \bbR^3$ be given by \eqref{eq:helixDiscreteDipoleFromBg}. Then
		\begin{equation*}
			e_\lambda \xrightarrow[\lambda \to 0]{} \dfrac{1}{2}h_0 \left[ ||\bP_{\perp} \bff||^2_{\LtwoOneD} - 2||\bP_{||} \bff||^2_{\LtwoOneD} \right],
		\end{equation*}
		where $h_0$ is the constant defined in \eqref{eq:helixH0Define}. 
	\end{theorem}
	
	\begin{proof}
		Since $T_\lambda$ is bounded and $\bff_\lambda \to \bff$, we have
		\begin{equation*}
		    \lim_{\lambda \to 0} (T_\lambda \bff_\lambda) = \lim_{\lambda \to 0} (T_\lambda \bff) + \lim_{\lambda \to 0}(T_\lambda (\bff_\lambda - \bff))
    			= \lim_{\lambda \to 0} (T_\lambda \bff) ,
		\end{equation*}
		where the decomposition of the limit into the sum of individual limits is true if the two limits, $\lim_{\lambda \to 0} (T_\lambda \bff)$ and $\lim_{\lambda \to 0}(T_\lambda (\bff_\lambda - \bff))$, individually exist. 
		This is indeed true: because $T_\lambda$ is bounded and $\bff_\lambda \to \bff$ in $L^2$, we have that $\lim_{\lambda \to 0}(T_\lambda (\bff_\lambda - \bff)) = 0$. 
		Then, to show that $\lim_{\lambda \to 0} (T_\lambda \bff)$ exist for $\bff\in \LtwoOneD$, we proceed as follows.
	
    	Let $\bff^k \in \CtestOneD$ be a sequence of functions such that $\bff^k \to \bff$. Using \sref{Proposition}{prop:helixTLambda}, we have 
    	\begin{equation}
    	    \begin{split}\label{eq:helixLimitTfShow}
    			\lim_{\lambda \to 0} (T_\lambda \bff) &= \lim_{k\to \infty} \lim_{\lambda \to 0}(T_\lambda \bff^k) + \lim_{k\to \infty} \lim_{\lambda \to 0} (T_\lambda (\bff - \bff^k))  = \lim_{k\to \infty} \lim_{\lambda \to 0}(T_\lambda \bff^k)  \\
    			&= \lim_{k\to \infty} \left( \bH_0 \bff^k  \right) = \bH_0 \bff  ,
    		\end{split}
    	\end{equation}
    	where $\bH_0 = \bH_0(s) = -h_0 (\bP_{\perp}(s) - 2\bP_{||}(s))$ (see \sref{Proposition}{prop:helixTLambda}).
    	
    	Using the expression in \eqref{eq:helixEnergyTLambda} for $e_\lambda$, we have
    	\begin{equation}
    	    \begin{split}
        		e_\lambda &= -\dfrac{1}{2} \Ltwodot{\bff_\lambda}{T_\lambda \bff_\lambda}_{\LtwoOneD} 
        		= -\dfrac{1}{2} \left[ \Ltwodot{\bff_\lambda - \bff}{T_\lambda \bff_\lambda}_{\LtwoOneD} + \Ltwodot{\bff}{T_\lambda \bff_\lambda}_{\LtwoOneD} \right]  \\
        		&= -\dfrac{1}{2} \left[ \Ltwodot{\bff_\lambda - \bff}{T_\lambda \bff_\lambda}_{\LtwoOneD} + \Ltwodot{\bff}{T_\lambda \bff}_{\LtwoOneD} + \Ltwodot{\bff}{T_\lambda ( \bff_\lambda - \bff)}_{\LtwoOneD} \right] .
    	    \end{split}
    	\end{equation}
    	The first and third terms are zero in the limit. Taking the limit of the remaining term and using \eqref{eq:helixLimitTfShow}, we have 
    	\begin{equation*}
    		\lim_{\lambda \to 0} e_\lambda = \lim_{\lambda \to 0} -\dfrac{1}{2} \Ltwodot{\bff_\lambda}{T_\lambda \bff_\lambda}_{\LtwoOneD} = \dfrac{1}{2} h_0 \left[ ||\bP_{\perp} \bff ||^2_{\LtwoOneD} - 2||\bP_{||} \bff ||^2_{\LtwoOneD}\right].
    	\end{equation*}
    	This completes the proof.
	\end{proof}
	
\begin{remark}
	The limiting energy only comprises of a local self-field energy. In the limit, any point on the helix sees a uniform 1-d system of dipole moments along the tangent line. Further, we see that both the normal components and the tangential component of the dipole moment contribute to the energy and electric field. 
        This is in contrast to \cite{chacouche2015ferromagnetic}, where the thin wire limit of the magnetostatic energy, obtained from dimensional reduction starting from a 3-d continuum, has contributions only from the normal component.
	    Heuristically, the dimension reduction starting from a 3-d continuum contains minimal information about the detailed atomic arrangements within the nanostructure, and hence does not capture the consequences of the helical geometry.
	    Because of its starting point in a 3-d continuum model, it is appropriate for thin objects that have all dimensions being much larger than the atomic lengthscale.
	    On the other hand, the discrete-to-continuum approach is appropriate for nanostructures wherein the thin dimension is comparable to the atomic lengthscale.
\end{remark}

	\subsection{Nanofilm with Constant Bending Curvature}
	\label{ss:thinfilm}

	Let $\calS = (-\thetab,\thetab)\times \bbR$ be the parametric space for a surface with a constant bending curvature $\kappa$. The map that takes a point in the parametric space to a unique point on the film is given by
	\begin{equation}\label{eq:filmBxBar}
		\bxbar(s_1, s_2) = \calR\bQ(s_1)\be_1 + s_2 \delta \be_3,
	\end{equation}
	where $\calR = 1/\kappa$ is the inverse of curvature, $\thetab > 0$ is the angular size of the film, and $\delta$ is the spacing in the flat direction. Here, $\kappa, \delta, \thetab$ are fixed parameters for a given film. 
	Here, $\bQ = \bQ(\theta)$ is the rotational tensor with the axis $\be_3$, see definition \eqref{eq:rotTensorQ}. The tangent vectors at $\bs := (s_1, s_2) \in S$ are
	\begin{equation}\label{eq:fimTangent}
	\bt_1(\bs) = \frac{\dmnosp \bxbar}{\dmnosp s_1} = \calR\bQ'(s_1) \be_1,  \qquad \bt_2(\bs) = \frac{\dmnosp \bxbar}{\dmnosp s_2} = \delta \be_3
	\end{equation}
	and the normal vector is 
	\begin{equation}\label{eq:fimNormal}
	\bn(\bs) = \bQ(s_1)\be_1.
	\end{equation}
	
	\subsubsection{Lattice Geometry and Dipole Moment}\label{sss:filmDipoleSystem}

	We consider a lattice embedded on the film $\bxbar$. We assume that the film is one lattice cell thick in the direction $\bn$ normal to the film. Suppose $\calL \subset \calS$ is the set of parametric coordinates of the discrete lattice sites. Let $\calL$ and the lattice cell $U$ (in the parametric space $\calS$) be given by
	\begin{equation}
    	\begin{split}\label{eq:latSitesL}
    		\calL &= \{\bs = (s_1, s_2) \in \calS;\, s_1 = i\theta_l, s_2 = j, \: i,j \in\bbZ\} = \left((-\thetab,\thetab)\cap \theta_l\bbZ\right) \times \bbZ \\
    		U(\bs) &= [s_1, s_1 + \theta_l) \times [s_2, s_2 + 1), \qquad \forall \bs \in \calL .
    	\end{split}
    \end{equation}
	Here, $\theta_l$ is the angular width of the lattice cell. We assume that $\theta_l$ is such that the set of sites in the angular direction, $(-\thetab, \thetab)\cap \theta_l \bbZ$, is not empty, and in fact is sufficiently large so that the continuum limit approximation of the energy density is justified. We assume that the lattice has unit thickness in the normal direction, and suppose that the film given by $\bxbar$ passes through the center of the lattice in the normal direction. Then, the unit cell for a given site $\bs \in \calL$ is $\bar{U}(\bs) = \{\bx;\, \bx = \bxbar(\bs') + t \bn(\bs'), \, \bs' \in U(\bs), \,t \in (-1/2, 1/2)\}$ (see \cref{fig:latticeThinFilm}). On the lattice $\calL$, we define a discrete system of dipole moments $\bd\colon \calL \to \bbR^3$. 
	As in the case of the helical nanotube, the lattice cells in real space are related by an isometric transformation, so the magnitudes of the dipoles at the lattice sites are equal, but they are oriented differently. In particular, we have
	\begin{equation}
	    \bfd(\bs + \br) = \bfQ(r_1) \bfd(\bs), \quad \bs \in \calL,
	\end{equation}
	where $\br = (r_1,r_2) \in \calL$ such that $\br + \bs \in \calL$ (i.e. all the translations within $\calL$). 
	We see that the dipole orientation depends only on the angular (first) parameter and is invariant with respect to the second parameter.

    \begin{figure}[!ht]
	\centering
    \begin{subfigure}[t]{.45\textwidth}
      \centering
      \includegraphics[width=.6\linewidth]{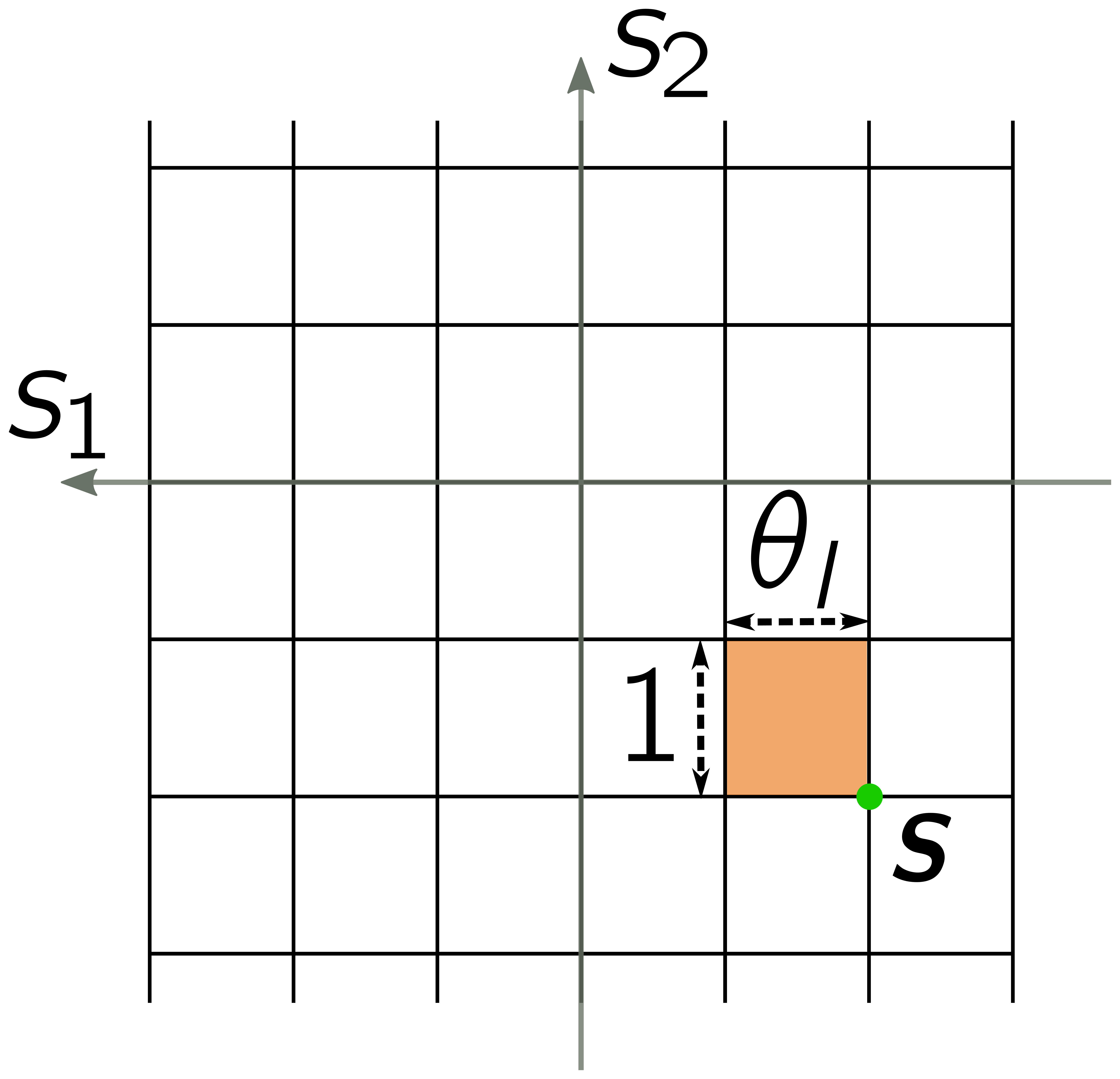}  
      \caption{}
    \end{subfigure}
    \begin{subfigure}[t]{.45\textwidth}
      \centering
      \includegraphics[width=.8\linewidth]{figure4b.pdf}  
      \caption{}
    \end{subfigure}
    \caption{Depiction of lattice sites in (a) parametric and (b) real space where the orange region marks a lattice cell in the parametric and real space.}
    \label{fig:latticeThinFilm}
    \end{figure}

	We next consider the scaling of the lattice by $\lambda >0$. The scaled lattice $\calL_\lambda$ and the associated lattice cell $U_\lambda$ are defined by the natural scaling of $\calL$ and $U$ as follows
	\begin{equation}
    	\begin{split}\label{eq:filmLLambda}
    		\calL_\lambda &= \{\bs \in \calS;\, s_1 = i \theta_l \lambda, s_2 = j\lambda, \, i, j \in \bbZ\} = \left((-\thetab,\thetab)\cap \lambda\theta_l\bbZ\right) \times \lambda\bbZ \\
    		U_\lambda(\bs) &= [s_1, s_1 + \lambda\theta_l) \times [s_2, s_2 + \lambda), \qquad \forall \bs \in \calL_\lambda .
	    \end{split}
    \end{equation}
	After scaling, the thickness of the lattice cell in the normal direction is $\lambda$ and the unit cell for $\bs\in \calL_\lambda$ is $\bar{U}_\lambda(\bs) = \{\bx;\, \bx = \bxbar(\bs') + t \bn(\bs'), \, \bs' \in U_\lambda(\bs), \, t \in (-\lambda/2, \lambda/2)\}$. We can show that the unit cell in the scaled lattice has volume $\lambda^3\calR\theta_l$. Let $\bdlambda\colon \calL_\lambda \to \R^3$ denote the discrete system of dipole moments associated with the scaled lattice, $\calL_\lambda$, and $\tilde{\bd}_\lambda \colon \calS \to \bbR^3$ denote the piecewise constant extension of $\bdlambda$ given by
	\begin{equation}\label{eq:filmExtendDipoleField}
	    \tilde{\bd}_\lambda(\bs) = \frac{\bdlambda(\ba)}{\lambda^3 \calR\theta_l},\qquad \forall \bs \in U_\lambda(\ba), \, \forall \ba \in \calL_\lambda.
	\end{equation}
	We are interested in the limit of the energy when $\tilde{\bd}_\lambda$ converges to $\bff$ in $L^2(\calS, \bbR^3)$. As in the case of the helix and following \cite{james1994internal}, we suppose that there exists a background dipole moment density field $\bff_\lambda \in L^2(\calS, \bbR^3)$ such that the dipole moment at site $\bs \in \calL_\lambda$ is given by
	\begin{equation}\label{eq:filmBgDipoleScale}
		\bd_\lambda(\bs) = \int_{-\lambda/2}^{\lambda/2} \left[\int_{U_\lambda(\bs)} \bff_\lambda(\bt) \calR \dm t_1 \dm t_2\right] \dm t_3 = \calR \lambda \int_{U_\lambda(\bs)} \bff_\lambda(\bt) \dm \bt ,
	\end{equation}
	where $\dmnosp \bt = \dmnosp t_1\dm t_2$ is the area measure (note that $\dmnosp \bt$ does not include $\calR$). The existence of one such background field $\bff_\lambda$ is evident: we can define $\bff_\lambda = \tilde{\bd}_\lambda$. Similar to the case of the helix, we have the following lemma that relates the convergence of the background dipole moment field and the piecewise constant extension.
	
	\begin{lemma}
	    Let $\bff_\lambda$, $\lambda > 0$, be a sequence of $L^2(\calS, \bbR^3)$ functions and let $\bff\in L^2(\calS, \bbR^3)$ be such that $\bff_\lambda \to \bff$ in $L^2(\calS, \bbR^3)$. Let $\bd_\lambda \colon \calL_\lambda \to \bbR^3$ be given by \eqref{eq:filmBgDipoleScale} and let $\tilde{\bd}_\lambda$ be a piecewise constant $L^2$ extension of $\bdlambda$ given by \eqref{eq:filmExtendDipoleField}. Then,  $\tilde{\bd}_\lambda \to \bff$ in $L^2(\calS, \bbR^3)$. 
	    
	    On the other hand, if $\bdlambda \colon \calL_\lambda \to \bbR^3$ is such that $\tilde{\bd}_\lambda \to \bff$ in $L^2(\calS, \bbR^3)$, then there exists a background field $\bff_\lambda \in L^2(\calS, \bbR^3)$ such that $\bdlambda$ is given by \eqref{eq:filmBgDipoleScale}. 
    \end{lemma}
    The proof follows directly from the proof of Theorem 4.1 of \cite{james1994internal}.
	
	\begin{remark}
	    As different unit cells are related by isometric transformations, the dipole moments in different unit cells are related by the rotational part of the isometric transformation.
	\end{remark}

	\begin{figure}[!ht]
	\centering
    \begin{subfigure}[t]{.45\textwidth}
      \centering
      \includegraphics[width=.75\linewidth]{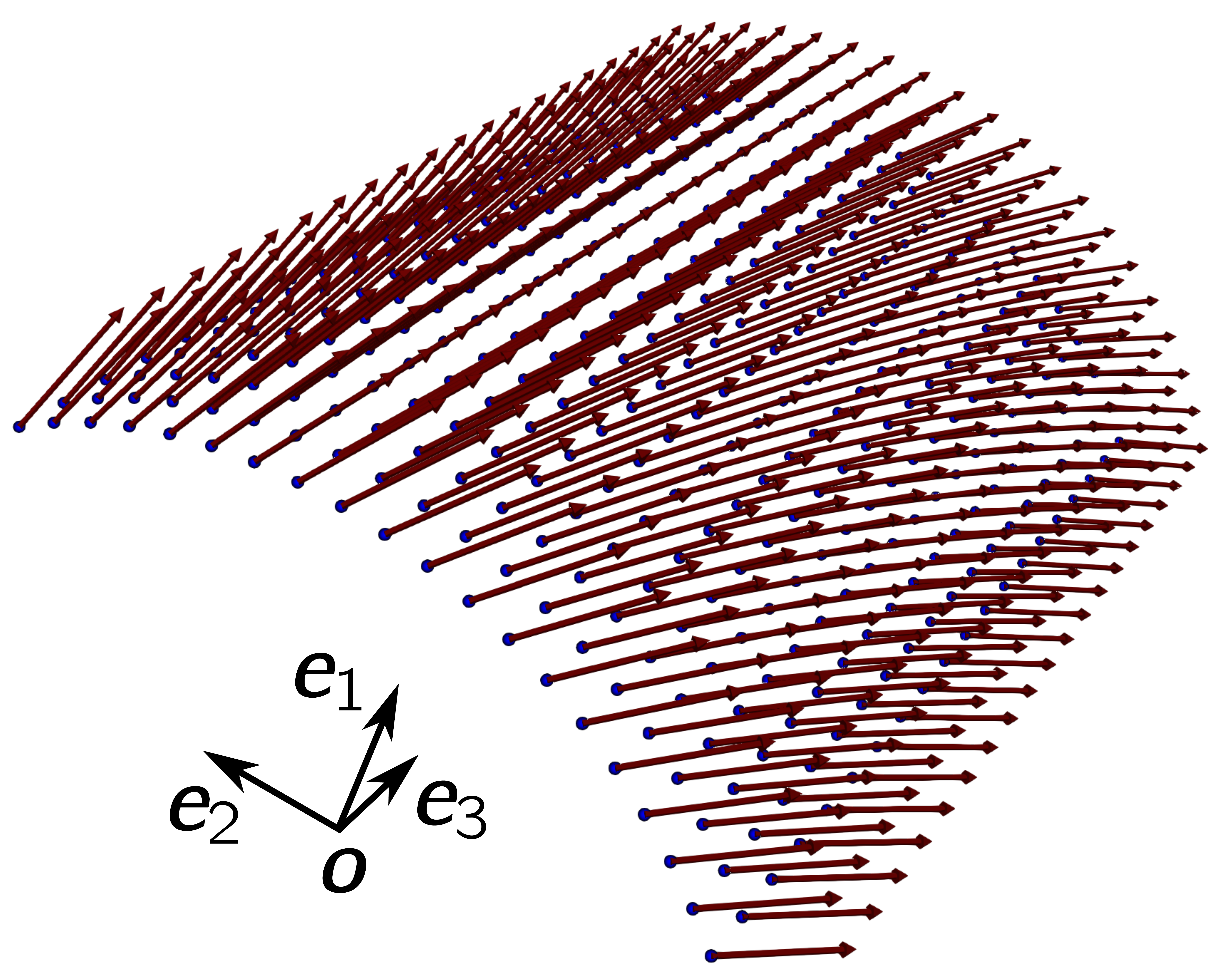}  
      \caption{}
    \end{subfigure}
    \begin{subfigure}[t]{.45\textwidth}
      \centering
      \includegraphics[width=.75\linewidth]{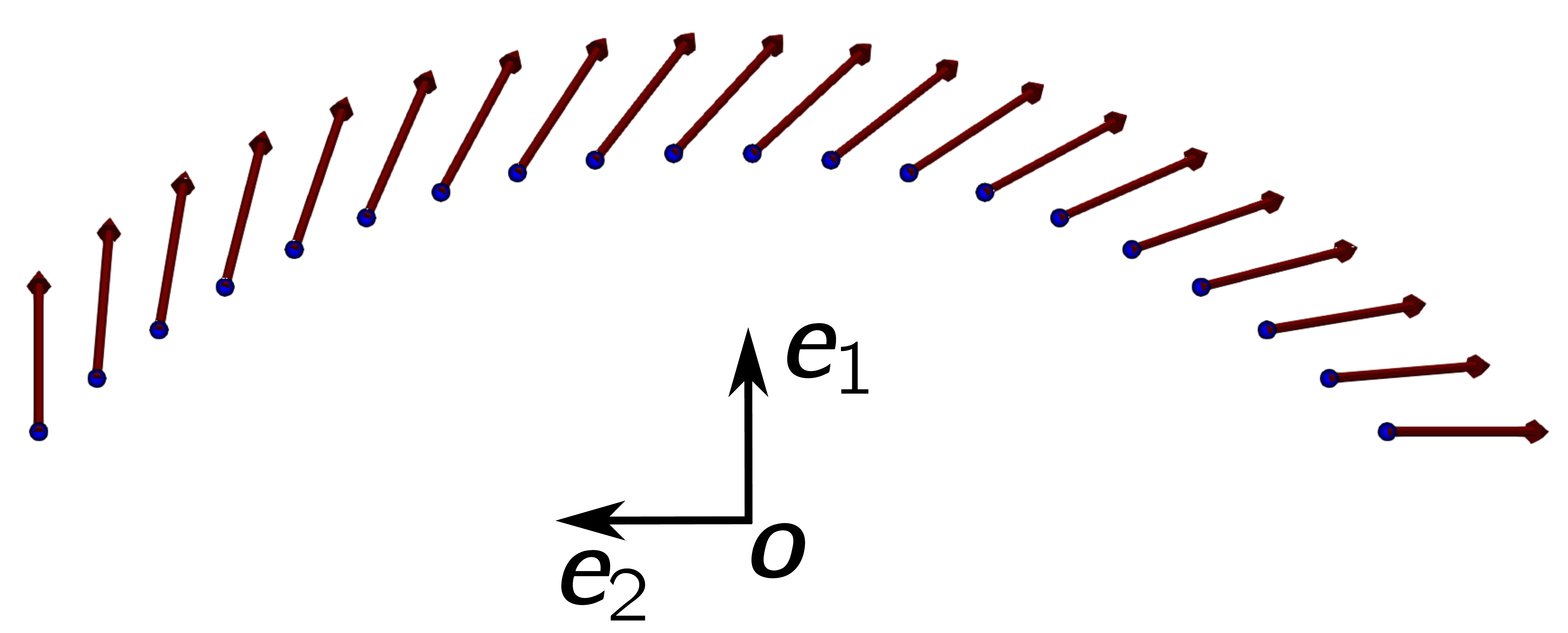}  
      \caption{}
    \end{subfigure}
    \caption{Discrete dipole moments on a nanofilm with uniform bending curvature. (a) and (b) show the view from different perspectives.}
    \label{fig:film}
    \end{figure}

	\subsubsection{Electrostatic Energy}\label{sss:energyFilm}
	As before, the energy associated to the system of dipole moments $\bd_\lambda$, for $\lambda >0$, is given by
	\begin{equation*}
		E_\lambda = -\frac{1}{2} \sum_{\substack{\bs, \bs' \in \calL_\lambda, \\
				\bs \neq \bs'}} \bd_\lambda(\bs) \cdot \bK(\bxbar(\bs') - \bxbar(\bs)) \bd_\lambda(\bs') = |(-\lambda/2, \lambda/2)| \hat{e}_\lambda, 
	\end{equation*}
	where $|(-\lambda/2, \lambda/2)| = \lambda$ is the thickness of the lattice in normal direction, and $\hat{e}_\lambda$ is the energy per unit length given by
	\begin{equation}\label{eq:filmEnergy}
		\hat{e}_\lambda = -\frac{1}{2\lambda} \sum_{\substack{\bs, \bs' \in \calL_\lambda, \\
				\bs \neq \bs'}} \bd_\lambda(\bs) \cdot \bK(\bxbar(\bs') - \bxbar(\bs)) \bd_\lambda(\bs') .
	\end{equation}
	For convenience, we normalize $\hat{e}_\lambda$ by $\calR\theta_l$, where $\calR\theta_l$ is independent of $\lambda$ and gives the size of the original lattice in the angular direction. We let
	\begin{equation}\label{eq:filmEnergyNormalized}
	    e_\lambda = \frac{\hat{e}_\lambda}{\calR \theta_l} \Rightarrow E_\lambda = \lambda (\calR\theta_l) e_\lambda .
	\end{equation}
	Substituting \eqref{eq:filmBgDipoleScale} and proceeding similar to the case of the helix, we can express $e_\lambda$ as
	\begin{equation}\label{eq:filmEnergyTLambda}
		e_\lambda = \Ltwodot{\bff_\lambda}{T_\lambda \bff_\lambda}_{L^2(\calS, \bbR^3)} , 
	\end{equation}
	where $T_\lambda \colon L^2(\calS,\bbR^3) \to L^2(\calS,\bbR^3)$ is the map defined as
	\begin{equation}\label{eq:filmTLambda}
		(T_\lambda \bff) (\bs) = \frac{\calR}{\theta_l}\lambda \int_{\calS} \bK_\lambda (\bs',\bs) \bff(\bs') \dm \bs'
	\end{equation}
	and $\bK_\lambda(\bs', \bs)$, for $\bs, \bs' \in \calS$, is the discrete dipole field kernel given by
	\begin{equation}\label{eq:filmKLambda}
		\bK_{\lambda}(\bs', \bs) = \sum_{\substack{\bu, \bv\in \Llambda,\\
				\bu \neq \bv}} \chi_{U_\lambda(\bv)}(\bs') \bK(\bxbar(\bv) - \bxbar(\bu)) \chi_{U_\lambda(\bu)}(\bs) .
	\end{equation}
	
	\paragraph{Scaling of  $\bK_\lambda$} 
	As in the case of the helix, it is convenient to first rescale the lattice $\calL_\lambda$ such that the lattice cell size is independent of $\lambda$ after rescaling, and define a new map on the rescaled lattice. This is considered next.
	
	Let $\calS_{1,\lambda} = (-\thetab/\lambda, \thetab/\lambda) \times \bbR$ so that $\bs\in S$ implies $\bs/\lambda \in \calS_{1,\lambda}$. We define a rescaled lattice $\calL_{1,\lambda}$ such that $\bs \in \calL_\lambda$ implies $\bs/\lambda \in \calL_{1,\lambda}$. It is given by
	\begin{equation}\label{eq:filmL1Lambda}
		\calL_{1,\lambda} = \{\bs \in \calS_{1,\lambda};\, s_1=i\theta_l, \, s_2 = j, \, i,j \in \bbZ \} = \left((-\thetab/\lambda, \thetab/\lambda)\cap \theta_l\bbZ\right)  \times \bbZ .
	\end{equation}
	The lattice cell for $\bs \in \calL_{1,\lambda}$ is given by $U_1(\bs)$, where $U_1(\bs)$ is defined in \eqref{eq:latSitesL}  (using $\lambda = 1$ in $U_\lambda$).  For $\ba, \bb \in \calS_{1,\lambda}$, we have
	\begin{equation}\label{eq:filmRelativeVecLambda}
		\bxbar(\lambda \ba) - \bxbar(\lambda \bb) = \lambda \left(\bxbar(\ba) - \bxbar(\bb) + \bA_\lambda(\ba, \bb) \be_1\right),
	\end{equation}
	where
	\begin{equation}\label{eq:filmALambda}
		\bA_\lambda(\ba, \bb) = \frac{\calR}{\lambda} \left[\bQ(\lambda a_1) - \bQ(\lambda b_1) - \lambda \bQ(a_1) + \lambda \bQ(b_1) \right].
	\end{equation}
	Keeping in mind these definitions, for $\bu \in \calL_{1,\lambda}$, we also note
	\begin{equation}
	    \begin{split}\label{eq:filmChiProp}
    		\chi_{U_\lambda(\lambda\bu)}(\bs) &= \begin{cases}
    			1 \qquad \text{if } \bs \in U_\lambda(\lambda \bu), \\
    			0 \qquad \text{otherwise}
    		\end{cases} = \quad \begin{cases}
    			1 \qquad \text{if } \bs /\lambda\in U_1(\bu), \\
    			0 \qquad \text{otherwise}
    		\end{cases} = \chi_{U_1(\bu)}(\bs/\lambda).
	    \end{split}
    \end{equation}
	Using the above relation and \eqref{eq:filmRelativeVecLambda}, we can show, for any $\bs,\bs' \in \calS$,
	\begin{equation*}
		\bK_\lambda(\bs', \bs) = \frac{1}{\lambda^3} \sum_{\substack{\bu,\bv \in \calL_{1,\lambda}, \\
				\bu \neq \bv}} \chi_{U_1(\bv)}(\bs'/\lambda) \bK(\bxbar(\bv) - \bxbar(\bu) + \bA_\lambda(\bv, \bu)\be_1) \chi_{U_{1}(\bu)}(\bs/\lambda) .
	\end{equation*}
	If we introduce the discrete dipole field kernel $\bK_{1,\lambda}(\bs',\bs)$, for $\bs,\bs'\in \calS_{1,\lambda}$, defined on $\calL_{1,\lambda}$ as:
	\begin{equation}\label{eq:filmK1Lambda}
		\bK_{1,\lambda}(\bs', \bs) = \sum_{\substack{\bu,\bv \in \calL_{1,\lambda}, \\
				\bu \neq \bv}} \chi_{U_1(\bv)}(\bs')  \bK(\bxbar(\bv) - \bxbar(\bu) + \bA_\lambda(\bv, \bu)\be_1) \chi_{U_{1}(\bu)}(\bs) ,
	\end{equation}
	we have shown that:
	\begin{equation*}
		\bK_\lambda(\bs', \bs) = \frac{1}{\lambda^3} \bK_{1,\lambda}(\bs'/\lambda, \bs/\lambda), \qquad \forall \bs,\bs' \in \calS.
	\end{equation*}
	
	\subsubsection{Limit of the Electrostatic Energy}
	
	In this section, we obtain the limit of the energy per unit length $e_\lambda$. The broad strategy is similar to the helical nanotube. We first show that the map $T_\lambda$ is bounded and obtain its limit. The continuum limit of the energy density $e_\lambda$ then follows easily.
	
	\paragraph{Limit of the Discrete Electric Field}
	Let $T_{1,\lambda}\colon L^2(\calS_{1,\lambda}, \bbR^3) \to L^2(\calS_{1,\lambda}, \bbR^3)$ be the map with kernel $\bK_{1,\lambda}$. For any function $\bff \in L^2(\calS_{1,\lambda}, \bbR^3)$, we have
	\begin{equation}\label{eq:filmT1Lambda}
		(T_{1,\lambda}\bff)(\bs) = \frac{\calR}{\theta_l}\int_{\calS_{1,\lambda}} \bK_{1,\lambda}(\bs', \bs) \bff(\bs') \dm \bs', \qquad \forall \bs \in \calS_{1,\lambda}.
	\end{equation}
	
	Let $\bH_\lambda = \bH_\lambda(\bs)$ be the zeroth order moment (with respect to the first argument) of kernel $\bK_{\lambda}$ given by
	\begin{equation}\label{eq:filmHLambda}
		\bH_{\lambda}(\bs) = \frac{\calR\lambda}{\theta_l} \int_{\bs' \in \calS} \bK_{\lambda}(\bs', \bs) \dm\bs', \quad \forall \bs \in \calS .
	\end{equation}
	
	We now state the limit result of $T_\lambda$.
	
	\begin{proposition}
		\label{prop:filmTLambda}
		Suppose $0 < \theta < \pi/4$. The maps $T_{1,\lambda}$ and $T_{\lambda}$ are bounded in $L^2$ for all $\lambda > 0$ and satisfy
		\begin{equation*}
			\norm{T_\lambda}_{\mathcal{L}(L^2(\calS, \bbR^3),L^2(\calS,\bbR^3))} = \norm{T_{1,\lambda}}_{\mathcal{L}(L^2(\calS_{1,\lambda}, \bbR^3),L^2(\calS_{1,\lambda},\bbR^3))} .
		\end{equation*}
		Further, for $\bff\in \CtestOneD$,
		\begin{equation*}
			(T_\lambda \bff)(\bs) \xrightarrow[\lambda \to 0]{} \bH_0(\bs)\bff(\bs),
		\end{equation*}
		pointwise, where $\bH_0(\bs)$, for $\bs \in \calS$, is given by
		\begin{equation}\label{eq:thinFilmH0Def}
			\bH_0(\bs) = \lim_{\lambda \to 0} \bH_\lambda(\bs) =  \calR \sum_{\substack{\bu = (u_1,u_2) \in \theta_l\bbZ \times \bbZ, \\ \bu \neq \bzero}}  \bK\left(u_1 \bt_1(\bs) + u_2 \bt_2(\bs) \right)
		\end{equation}
		and $\bt_i(\bs) = \frac{\dmnosp \bxbar(\bs)}{\dmnosp s_i}$, $i=1,2$, are the tangent vectors to the film.
	\end{proposition}
	
	We provide the proof of \sref{Proposition}{prop:filmTLambda} in \autoref{ss:filmProof}. Based on the proposition above, we state the main result for the thin film.
	
	\paragraph{Limit of the Energy}
	
	\begin{theorem}
		\label{thm:filmELambda}
		Let $\bff_\lambda \in L^2(\calS, \bbR^3)$ be a sequence of functions for $\lambda > 0$ with $\bff \in L^2(\calS, \bbR^3)$ such that $\bff_\lambda \to \bff$ in $L^2(\calS, \bbR^3)$. Let the system of dipole moments $\bdlambda\colon \calL_\lambda \to \bbR^3$ be given by \eqref{eq:filmBgDipoleScale}. Let $e_\lambda$, given by \eqref{eq:filmEnergyTLambda}, be the energy per unit length normalized by $\calR \theta_l$. Then
		\begin{equation*}
			e_\lambda \xrightarrow[\lambda \to 0]{} -\dfrac{1}{2} \Ltwodot{\bff}{\bH_0\bff}_{ L^2(\calS, \bbR^3)} ,
		\end{equation*}
		where $\bH_0 = \bH_0(\bs)$ is defined in \sref{Proposition}{prop:filmTLambda}, see \eqref{eq:thinFilmH0Def}. 
	\end{theorem}
	
	The proof of \autoref{thm:filmELambda} follows from the proof of \autoref{thm:helixELambda} and using \sref{Proposition}{prop:filmTLambda}. 
	
\begin{remark}
	Note that, for $\bs\in \calS$,
	\begin{equation}
		\bH_0(\bs) =  \bQ(s_1) \bH_0(\bzero) \bQ(-s_1) .
	\end{equation}
	Thus, if the limiting dipole moment field $\bff$ is uniform in the $\be_3$ direction, the electric field $\bH_0(\bs)\bff(\bs)$ will be independent of the $\be_3$-coordinate. It is easy to see from the expression of $\bH_0$ that both the normal component and the tangential components of the dipole field contribute to the electric field and energy. 
	This is in contrast to \cite{gioia1997micromagnetics}, where the thin film limit of the magnetostatic energy, obtained from dimensional reduction starting from a 3-d continuum, has contributions only from the normal component.
	    However, it is consistent with the result for the helical nanostructure studied in this paper, in that the components of the dipole field aligned with the thin direction contribute to the limit energy.
	    As we argued there, the dimensional reduction from a 3-d continuum model is appropriate for thin films that have thickness that is much larger than the atomic lengthscale, whereas the discrete-to-continuum approach is appropriate for nanostructures that have thickness that is comparable to the atomic lengthscale.
\end{remark}

	\section{Proof of Assertions}\label{s:proof}
	
	\subsection{Helical Nanotube}\label{ss:helixProof}
	
	In this section, we prove \sref{Proposition}{prop:helixTLambda}. First, we collect some important results, and then show that $T_\lambda$ is bounded and extends from $\bff\in C^\infty_0(\bbR, \bbR^3)$ to $L^2(\bbR, \bbR^3)$. We then obtain the limit of the map $T_\lambda$. 
	
	\begin{lemma}
		
		\begin{itemize}
			\item[1.] For any $a, b \in \bbR$, 
			\begin{equation}\label{eq:rotIdent1}
				\bar{\bx}(b) - \bar{\bx}(a) = \bQ(a\theta) [(\bQ((b - a)\theta) - \bI) \be_1 + \delta (b - a)\be_3] ,
			\end{equation}
			where $\bxbar$ is the map \eqref{eq:helixMap}, $\bQ$ is the rotational tensor \eqref{eq:rotTensorQ}, $\theta$ and $\delta$ define the helix.
			
			\item[2.] For any $\theta \in (0, \pi)$,
			\begin{equation}\label{eq:helixMapIdent1}
				\delta \leq \min_{a,b\in \calL_1, a\neq b} |\bar{\bx}(b) - \bar{\bx}(a)| ,
			\end{equation}
			where $\calL_1$ is $\calL_\lambda= \lambda\bbZ$ for $\lambda = 1$.
			
			\item[3.] For any $a,b\in \calL_1$ and $\lambda>0$,
			\begin{equation}\label{eq:helixMapIdent2}
				\delta |a - b| \leq  |\bar{\bx}(a) - \bar{\bx}(b) + \bA_\lambda(a,b)\be_1| ,
			\end{equation}
			where $\bA_\lambda(a,b)$ is given by
			\begin{equation*}
				\bA_\lambda(a,b) = \frac{\bQ(\lambda a\theta) - \bQ(\lambda b\theta) - (\lambda \bQ(a\theta) - \lambda \bQ(b\theta))}{\lambda} .
			\end{equation*}
			
			\item[4.] For any $s,s'\in \bbR$ such that $|s-s'| \geq 1$, suppose $a, b\in \calL_1$ are such that $s\in [a, a+1), s' \in [b, b+1)$, then
			\begin{equation}\label{eq:helixLatticeCellIdent1}
				\frac{|s - s'|}{|a - b|} < 3 .
			\end{equation}
		\end{itemize}
		
	\end{lemma}
	
\begin{proof}
    \begin{enumerate}

        \item 
		For any $\alpha, \beta \in \bbR$, we have the identities
		\begin{equation}\label{eq:rotIdent}
			\bQ^T(\alpha) = \bQ(-\alpha), \qquad \bQ(\alpha)\bQ(\beta) = \bQ(\alpha + \beta), \qquad \bQ(\alpha)\be_3 = \be_3,
		\end{equation}
		where the last relation shows that $\be_3$ is the axis of $\bQ$. By noting the definition of $\bxbar$ in \eqref{eq:helixMap} and using the identities above, \eqref{eq:rotIdent1} follows.
		
		\item
		To show \eqref{eq:helixMapIdent1}, we use \eqref{eq:rotIdent1} to get
		\begin{equation*}
			|\bxbar(b) - \bxbar(a)|^2 = |(\bQ((b - a)\theta) - \bI) \be_1|^2 + \delta^2 |b - a|^2 \geq \delta^2 |b-a|^2 \geq \delta^2,
		\end{equation*}
		where we used the fact that $|b - a| \geq 1$ for $a, b\in \calL_1, a \neq b$.
		
		\item
		To show \eqref{eq:helixMapIdent2}, we substitute the definition of $\bA_\lambda$ to get
		\begin{equation}
			\bxbar(a) - \bxbar(b) + \bA_\lambda(a,b) \be_1 =  \frac{\bQ(\lambda a\theta ) - \bQ(\lambda b\theta)}{\lambda} \be_1 + (a-b)\delta \be_3 .
		\end{equation}
		Since $\bQ(\alpha)\be_1$ is orthogonal to $\be_3$ for any $\alpha$, we have
		\begin{equation*}
			|\bxbar(a) - \bxbar(b) + \bA_\lambda(a,b) \be_1| \geq \delta |a - b| .
		\end{equation*}
		
		\item
		To show \eqref{eq:helixLatticeCellIdent1}, we note that for $s,s' \in \bbR$ such that $|s-s'|\geq 1$ with $a,b\in \calL_1$ and $s\in [a, a+1), s'\in [b, b+1)$, we can write $s  = a + \Delta s$ and $s' = b + \Delta s'$ with $0\leq \Delta s, \Delta s'  < 1$. Thus
		\begin{equation*}
			\frac{|s - s'|}{|a - b|} = \frac{|a - b + (\Delta s - \Delta s')|}{|a - b|} \leq \frac{|a-b| + |\Delta s - \Delta s'|}{|a - b|} < 1 + \frac{ 2}{|a - b|} \leq 3,
		\end{equation*}
		where in the last step we used the fact that $|a-b|\geq 1$ for $a,b\in \calL_1, a\neq b$ (which is ensured when $|s-s'|\geq 1$).
	 
    \end{enumerate}
    This completes the proof.
\end{proof}

	\subsubsection{Boundedness}\label{sss:helixProofBdd}
	
	We next show that $T_\lambda$ is a bounded map.
	Let $S_\lambda \colon \LtwoOneD \to \LtwoOneD$ be an isometry defined as
	\begin{equation*}
		(S_\lambda \bff)(s) := \lambda^{1/2} \bff(\lambda s) .
	\end{equation*}
	It is easy to see that $||S_\lambda\bff||_{\LtwoOneD} = ||\bff||_{\LtwoOneD}$. The inverse of $S_\lambda$ is given by
	\begin{equation}\label{eq:SlamInvHelix}
		(S^{-1}_\lambda \bff) (s) = \lambda^{-1/2} \bff(s/\lambda) .
	\end{equation}
	
	Using $S_\lambda$, we can show -- noting the definition of $T_\lambda$ in \eqref{eq:helixTLambda} --  for $\bff\in \LtwoOneD$,
	\begin{equation}
	    \begin{split}
    		(T_\lambda \bff) (s) &= \lambda^2 \int_{\bbR} \bK_{\lambda}(s',s) \bff(s') \dm s' 
    		= \lambda^2 \int_{\bbR} \dfrac{1}{\lambda^3} \bK_{1,\lambda}(s'/\lambda, s/\lambda) \left(\lambda^{-1/2}(S_\lambda \bff)(s'/\lambda)\right) \dm s'  \\
    		&= \lambda^{-3/2} \int_{\bbR} \bK_{1,\lambda}(s',s/\lambda) (S_\lambda \bff) (s') \lambda \dm s' = \lambda^{-1/2} (T_{1,\lambda}(S_\lambda \bff))(s/\lambda)  \\
    		&= (S_\lambda^{-1} T_{1,\lambda} S_\lambda \bff)(s) ,
	    \end{split}
    \end{equation}
	where we used a change of variables and \eqref{eq:SlamInvHelix}. It follows from the above equation that
	\begin{equation}
	    \begin{split}\label{eq:helixTlamT1lamEquality}
    		\norm{T_\lambda}_{\mathcal{L}(L^2,L^2)} &= \sup_{\norm{\bff} \neq 0} \dfrac{\norm{T_{\lambda} \bff}_{\LtwoOneD}}{\norm{\bff}_{\LtwoOneD}} = \sup_{\norm{\bff} \neq 0} \dfrac{\norm{S_{\lambda}^{-1}(T_{1,\lambda} S_{\lambda} \bff)}_{\LtwoOneD}}{\norm{\bff}_{\LtwoOneD}} \\
    		&= \sup_{\norm{\bff} \neq 0} \dfrac{\norm{T_{1,\lambda} (S_{\lambda} \bff)}_{\LtwoOneD}}{\norm{\bff}_{\LtwoOneD}} = \sup_{\norm{S_{\lambda} \bff} \neq 0} \dfrac{\norm{T_{1,\lambda} (S_{\lambda} \bff)}_{\LtwoOneD}}{\norm{S_{\lambda} \bff}_{\LtwoOneD}}  \\
    		&= \norm{T_{1,\lambda}}_{\mathcal{L}(L^2,L^2)} ,
	    \end{split}
    \end{equation}
	    where we have used that $||\bff||_{L^2(\R, \R^3)} = ||S_\lambda \bff||_{L^2(\R, \R^3)}$.
	This completes the proof of \eqref{eq:helixTLambdaEquiv} in \sref{Proposition}{prop:helixTLambda}. Next, we show that $T_{1,\lambda}$ is a bounded map to prove the boundedness of $T_\lambda$. We first analyze the discrete dipole field kernel $\bK_{1,\lambda}$, which is defined as
	\begin{equation}\label{eq:helixK1LambdaRedef}
		\bK_{1,\lambda}(s',s) = \sum_{\substack{u,v\in \calL_1,\\
				u \neq v}} \chi_{U_1(v)}(s') \bK(\bxbar(v) - \bxbar(u) + \bA_{\lambda}(v,u) \be_1 ) \chi_{U_1(u)}(s) ,
	\end{equation}
	where $U_\lambda(s) = [s, s+\lambda)$ for $s\in \calL_\lambda$, and $\bA_\lambda(a,b)$ is given by \eqref{eq:helixALambda}. 
	
	Consider some typical $s,s'\in \bbR$ and the corresponding $a,b\in \calL_1$ such that $s\in [a,a+1), s'\in [b, b+1)$. From \eqref{eq:helixK1LambdaRedef}, we have, for all $s,s'\in \bbR$ such that $|s-s'| < 1$,
	\begin{itemize}
		\item If $a = b$, then $\bK_{1,\lambda}(s,s') = \bzero$.
		\item If $a\neq b$, then from \eqref{eq:helixMapIdent1}, we have $$
			|\bK_{1,\lambda}(s,s')| \leq \sqrt{6}/(4\pi\delta^3)$$ using $|\bA\ba| \leq |\bA| \, |\ba|$ and $|\bI - 3(\bx/|\bx|) \dyad (\bx/|\bx|)| \leq \sqrt{6}$, $\forall \bx \neq \bzero$ .
	\end{itemize}
	Combining the two cases above, $|\bK_{1,\lambda}(s,s')| \leq \sqrt{6}/(4\pi\delta^3)$.
	
	We now consider the case when $|s - s'| \geq 1$. Noting that for this case, $a\neq b$. We proceed as follows
	\begin{equation}
	    \begin{split}
    		|\bK_{1,\lambda}(s,s')| &\leq \dfrac{\sqrt{6}}{4\pi|s-s'|^3} \dfrac{|s-s'|^3}{|a-b|^3} \dfrac{|a-b|^3}{|\bxbar(a) - \bxbar(b) + \bA_{\lambda}(a,b) \be_1|^3}  \\
    		&\leq  \dfrac{\sqrt{6}}{4\pi|s-s'|^3}  3^3 \dfrac{|a-b|^3}{\delta^3 |a - b|^3} = \frac{3^3 \sqrt{6}}{4\pi\delta^3} \frac{1}{|s-s'|^3},
	    \end{split}
    \end{equation}
	where we used the bounds \eqref{eq:helixMapIdent2} and \eqref{eq:helixLatticeCellIdent1}. Combining  the above bound for $|s-s'|\geq 1$ with the bound for $|s-s'| <1$, and renaming the constants, we can write
	\begin{equation}\label{eq:helixBddK2}
		|\bK_{1,\lambda}(s,s')| \leq  \dfrac{C_1}{C_2 + |s-s'|^3} .
	\end{equation}
	
    	Next, note that, since the kernel $\bK_{1,\lambda}$ satisfies \eqref{eq:helixBddK2}, we have
    	\begin{equation}\label{eq:helixK1lamIntegBd}
    	    \int_{\bbR} |\bK_{1,\lambda}(s', s)| \dm s' \leq C_3, \quad \int_{\bbR} |\bK_{1,\lambda}(s', s)| \dm s \leq C_3,
    	\end{equation}
    	for some fixed $C_3 < \infty$ independent of $\lambda$. Using the above bound, we can show that, for all  $\bff\in C^\infty_0(\bbR, \bbR^3)$,
    	\begin{equation}\label{eq:helixT1lambdaBdd}
    	    ||T_{1,\lambda}\bff||_{L^2(\bbR, \bbR^3)} \leq C_3 ||\bff||_{L^2(\bbR, \bbR^3)},
    	\end{equation}
    	which establishes that $T_{1,\lambda}$ is a bounded linear map on $C^\infty_0(\bbR, \bbR^3)$. Since $C^\infty_0(\bbR, \bbR^3)$ is dense in $L^2(\bbR, \bbR^3)$, it follows that $T_{1, \lambda}$ is also bounded in $L^2(\bbR, \bbR^3)$, and extends as a bounded linear map from $C^\infty_0(\bbR, \bbR^3)$ to $L^2(\bbR, \bbR^3)$. This argument together with \eqref{eq:helixTlamT1lamEquality} completes the proof of boundedness of maps $T_\lambda$ and $T_{1,\lambda}$. 
    	Now, it remains to show \eqref{eq:helixT1lambdaBdd} for $\bff\in C^\infty_0(\bbR, \bbR^3)$. Let $\bff\in C^\infty_0(\bbR, \bbR^3)$ and proceed as follows:
    	\begin{equation}\label{eq:helixT1lambdaBddPf}
        	\begin{split}
        	    ||T_{1,\lambda}\bff||_{L^2(\bbR, \bbR^3)}^2 &= \int_{\R} \levert (T_{1,\lambda} \bff)(s) \rivert^2 \dm s = \int_{\R} \left( \int_{\R} \bK_{1,\lambda}(s', s) \bff(s') \dm s' \right)^2 \dm s  \\
        	    &= \int_{\R} \left[ \int_{\R} \int_{\R} |\bK_{1,\lambda}(s', s)|\, |\bK_{1,\lambda}(t', s)|\, |\bff(s')| \, |\bff(t')|\, \dm s' \dm t' \right] \dm s  \\
        	    &\leq \int_{\R} \left[ \int_{\R} \int_{\R}|\bK_{1,\lambda}(s', s)|\, |\bK_{1,\lambda}(t', s)|\,\left(\frac{|\bff(s')|^2}{2} + \frac{|\bff(t')|^2}{2}\right) \dm s' \dm t'  \right] \dm s  \\
        	    &= \int_{\R} \left[ \frac{1}{2} 2 \left(\int_{\R} |\bK_{1,\lambda}(s', s)| \dm s' \right) \left(\int_{\R} |\bK_{1,\lambda}(t', s)|\,|\bff(t')|^2 \dm t' \right) \right] \dm s  \\
        	    &\leq C_3 \underbrace{\int_{\R} \left[\int_{\R} |\bK_{1,\lambda}(t', s)|\,|\bff(t')|^2 \dm t' \right] \dm s}_{=: I}, 
    	    \end{split}
    	\end{equation}
    	where, in the third line we have used that $|\bff(s')| \, |\bff(t')| \leq |\bff(s')|^2/2 + |\bff(t')|^2/2$; in the fourth line, we have used symmetry to extract a factor of $2$; in the last line, we have used the bound \eqref{eq:helixK1lamIntegBd}.
        Since $\bff\in C^\infty_0(\bbR, \bbR^3)$, there exist $R>0$ such that the support of $\bff$ is a subset of $(-R, R)$. So, for any $\rho > 0$, we have
        \begin{equation}
        	\begin{split}
                I_\rho &:= \int_{-\rho}^{\rho} \left[\int_{\R} |\bK_{1,\lambda}(t', s)|\,|\bff(t')|^2 \dm t' \right] \dm s = \int_{-\rho}^{\rho} \left[\int_{-R}^{R} |\bK_{1,\lambda}(t', s)|\,|\bff(t')|^2 \dm t' \right] \dm s  \\
                &= \int_{-R}^{R} |\bff(t')|^2 \left( \int_{-\rho}^{\rho} |\bK_{1,\lambda}(t', s)| \dm s \right) \dm t' \leq C_3 ||\bff||_{L^2(\bbR, \bbR^3)}^2,
            \end{split}
        \end{equation}
        where we have applied Fubini's theorem to switch the order of integration -- this is allowed because $|\bK_{1,\lambda}(t', s)|\,|\bff(t')|^2$ is integrable in $(s,t') \in (-\rho, \rho)\times (-R, R)$ -- and then used \eqref{eq:helixK1lamIntegBd}. 
        Thus, we have shown that $I_\rho$ is bounded, with the bound independent of $\rho$. 
        This, together with the fact that $I_\rho$ is monotonically increasing with $\rho$, shows that $\lim_{\rho \to \infty} I_\rho$ exists and is equal to $I$, where $I$ is defined in \eqref{eq:helixT1lambdaBddPf}. 
        The limit $I$ is also bounded, i.e., $I \leq C_3 ||\bff||_{L^2(\bbR, \bbR^3)}^2$. Combining this observation with \eqref{eq:helixT1lambdaBddPf} completes the proof of \eqref{eq:helixT1lambdaBdd}.

	\subsubsection{Limit of the Map $T_\lambda$}\label{ss:limit T helix}
	
	Let $\bff \in \CtestOneD$. Write $T_\lambda \bff$ as:
	\begin{equation}
	\label{eq:helixTLambdaDecomp} 
		T_\lambda(\bff)(s) 
		= \lambda^2\int_{\bbR} \bK_{\lambda}(s', s) \bff(s') \dm s'
		= \underbrace{\left[\lambda^2\int_{\bbR} \bK_{\lambda}(s', s) \dm s'\right]}_{=: \bH_\lambda(s)} \bff(s) + \lambda^2 \int_{\bbR} \bK_{\lambda}(s', s) (\bff(s') - \bff(s)) \dm s' .
	\end{equation}
	The second term above is zero in the limit $\lambda \to 0$. To see this, we first obtain two useful inequalities. For any $R> 0$, we have, using the bound on $|\bK_{1,\lambda}|$ from \eqref{eq:helixBddK2}, 
	\begin{equation}
    \begin{split}
    \int_{\abs{s-s'} \leq R\lambda} |\bK_{1,\lambda}(s'/\lambda,s/\lambda) | \dm s' \leq \int_{\abs{s-s'} \leq R\lambda} \dfrac{C_1}{C_2 + |s-s'|^3/\lambda^3} \dm s' 
	    \leq \frac{C_1}{C_2}\int_{\abs{s-s'} \leq R\lambda}  \dm s' = \frac{C_1}{C_2} R\lambda \label{eq:helixUpperBdd2}
    \end{split}
    \end{equation}
    and
    \begin{equation}
    \begin{split}
		&\dfrac{1}{\lambda} \int_{\abs{s-s'} \geq R\lambda} |\bK_{1,\lambda}(s'/\lambda,s/\lambda) | \dm s'  \leq \dfrac{1}{\lambda} \int_{\abs{s-s'} \geq R\lambda} \dfrac{C_1}{C_2 + \abs{s/\lambda - s'/\lambda}^3} \dm s' = \int_{\abs{t} \geq R } \dfrac{C_1}{C_2 + \abs{t}^3} \dm t,  \label{eq:helixUpperBdd3}
	 \end{split}
    \end{equation}
    where the last equality in the equation above follows from the change of variables $t = (s' - s)/\lambda$. 
    Next, noting that $\lambda^2 \bK_{\lambda}(s', s) = \frac{1}{\lambda}\bK_{1,\lambda}(s'/\lambda, s/\lambda)$, we obtain the following bound on the second term of \eqref{eq:helixTLambdaDecomp}:
	\begin{equation}
	\label{eq:helixUpperBdd1}
    \begin{split}
		&\levert \lambda^2\int_{\bbR} \bK_{\lambda}(s', s) (\bff(s') - \bff(s)) \dm s' \rivert \\
		&= \levert \dfrac{1}{\lambda} \int_{\bbR} \bK_{1,\lambda}(s'/\lambda,s/\lambda) (\bff(s') - \bff(s)) \dm s' \rivert \\
		&= \left\vert \dfrac{1}{\lambda} \int_{\abs{s-s'} \geq R\lambda} \bK_{1,\lambda}(s'/\lambda,s/\lambda) \bff(s') \dm s' - \dfrac{1}{\lambda} \int_{\abs{s-s'} \geq R\lambda} \bK_{1,\lambda}(s'/\lambda,s/\lambda) \dm s' \bff(s) \right.  \\
		&\quad \left. + \dfrac{1}{\lambda} \int_{\abs{s-s'} \leq R\lambda} \bK_{1,\lambda}(s'/\lambda,s/\lambda) (\bff(s') - \bff(s)) \dm s'  \right\vert \\
		&\leq \left\vert \dfrac{1}{\lambda} \int_{\abs{s-s'} \geq R\lambda} \bK_{1,\lambda}(s'/\lambda,s/\lambda) \bff(s') \dm s' \right\vert + 
		\left\vert \dfrac{1}{\lambda} \int_{\abs{s-s'} \geq R\lambda} \bK_{1,\lambda}(s'/\lambda,s/\lambda) \dm s' \bff(s)\right\vert \\
		&\quad + \left\vert \dfrac{1}{\lambda} \int_{\abs{s-s'} \leq R\lambda} \bK_{1,\lambda}(s'/\lambda,s/\lambda) (\bff(s') - \bff(s)) \dm s'\right\vert \\
		&\leq \sup_{t} |\bff'(t)| \frac{R\lambda}{\lambda} \int_{\abs{s-s'} \leq R\lambda} |\bK_{1,\lambda}(s'/\lambda,s/\lambda) | \dm s'  + 2 \sup_{t} |\bff (t)| \frac{1}{\lambda} \int_{\abs{s-s'} \geq R\lambda} | \bK_{1,\lambda}(s'/\lambda,s/\lambda) | \dm s' \\
		&\leq \sup_{t} |\bff'(t)| \frac{C_1 R^2}{C_2}\lambda  + 2 \sup_{t} |\bff (t)| \int_{\abs{t} \geq R } \dfrac{C_1}{C_2 + \abs{t}^3} \dm t, 
	\end{split}
    \end{equation}
	where we have used the fact that $\bff\in \CtestOneD$, and therefore, $| \bff(s')| \leq \sup_t |\bff(t)|$, and $|\bff(s') - \bff(s)| \leq R\lambda \sup_t |\bff'(t)|$ for $s' \in \{t: |s - t| \leq R \lambda\}$. We have also used the inequalities \eqref{eq:helixUpperBdd2} and \eqref{eq:helixUpperBdd3} in the last step. 
	
	We note that the final inequality in \eqref{eq:helixUpperBdd1} is true for any $R>0$.
	Further, the two terms on the right side have a limit as $\lambda \to 0$  -- the second term clearly is independent of $\lambda$ -- for any $R> 0$.
	Therefore, the limit $\lambda \to 0$ of the two terms individually is equal to the limit of the sum, keeping $R$ fixed. 
	Thus, taking the limit $\lambda \to 0$ in \eqref{eq:helixUpperBdd1}, we have
	\begin{equation}
	    \begin{split}
	        &\lim_{\lambda \to 0} \, \levert \lambda^2\int_{\bbR} \bK_{\lambda}(s', s) (\bff(s') - \bff(s)) \dm s' \rivert \\
	        &\leq \lim_{\lambda \to 0}\left[ \sup_{t} |\bff'(t)| \frac{C_1}{C_2} R\lambda \right] + \lim_{\lambda \to 0}\left[ 2 \sup_{t} |\bff (t)| \int_{\abs{t} \geq R } \dfrac{C_1}{C_2 + \abs{t}^3} \dm t \right] \\
	        &= 2 \sup_{t} |\bff (t)| \int_{\abs{t} \geq R } \dfrac{C_1}{C_2 + \abs{t}^3} \dm t,
	    \end{split}
	\end{equation}
	for any $R > 0$. 
	Since the bound above is true for any $R> 0$, and the left side is independent of $R$, we can take the limit $R \to \infty$, where the limit exists and is $0$ for the right side, to get
	\begin{equation}
	    \begin{split}
	        \lim_{\lambda \to 0} \, \levert \lambda^2\int_{\bbR} \bK_{\lambda}(s', s) (\bff(s') - \bff(s)) \dm s' \rivert
	        \leq \lim_{R \to \infty} \, 2 \sup_{t} |\bff (t)| \int_{\abs{t} \geq R } \dfrac{C_1}{C_2 + \abs{t}^3} \dm t = 0. \label{eq:helixUpperBdd4}
	    \end{split}
	\end{equation}

	Thus, we have from \eqref{eq:helixTLambdaDecomp} that
	\begin{equation}
		\lim_{\lambda \to 0} T_{\lambda} (\bff)(s) = \left[ \lim_{\lambda \to 0} \bH_{\lambda}(s) \right] \bff(s). \label{eq:helixTLambdaLimitInter}
	\end{equation} 
	
	We next compute the limit of $\bH_\lambda(s)$. Fix $s\in \bbR$ and suppose $a\in \calL_\lambda$ such that $s\in U_\lambda(a)$. Using the definition of $\bK_\lambda(s',s)$, we have
	\begin{equation}
    		\bH_\lambda(s) = \lambda^2 \int_{\bbR} \bK(s',s) \dm s' = \lambda^2 \sum_{\substack{u\in \calL_1,\\
    				u\neq a}} \bK(\bxbar(u) - \bxbar(a)) \int_{U_\lambda(u)} \dm t  
    				= \lambda^3 \sum_{\substack{u\in \lambda\bbZ,\\
    				u\neq a}} \bK(\bxbar(u) - \bxbar(a)) .
    \end{equation}
	From \eqref{eq:rotIdent}, we have $\bxbar(u) - \bxbar(a) = \bQ(a\theta) (( \bQ((u-a)\theta) - \bI) \be_1 + (u-a)\delta \be_3)$. Using the identity $\bK(\bQ \bx) = \bQ \bK(\bx) \bQ^T$ and $\bK(\lambda \bx) = \bK(\bx)/\lambda^3$, we get
	\begin{equation}
	    \begin{split}
    		\bH_\lambda(s) &= \bQ(a\theta) \left[ \sum_{u\in \lambda\bbZ, u\neq a} \bK((\bQ((u-a)\theta) - \bI)/\lambda \be_1 + (u-a)\delta/\lambda \be_3) \right] \bQ(-a\theta)  \\
    		&= \bQ(a\theta) \left[ \sum_{i\in \bbZ, i\neq 0} \bK((\bQ(i\lambda\theta) - \bI )/\lambda \be_1 + i\delta \be_3) \right] \bQ(-a\theta) ,
    	\end{split}
    \end{equation}
	where we have changed variables $i = (u-a)/\lambda$. Note that $a\in \lambda\bbZ$, and, therefore, $(u-a)\in \lambda\bbZ$ for $u \in \lambda\bbZ$, which implies $ i \in \bbZ$. Since $s$ is related to $a$ by $s\in U_\lambda(a)$, we have $a \to s$ in the limit $\lambda \to 0$. Therefore, we get
	\begin{equation*}
		\bH_0(s) := \lim_{\lambda \to 0} \bH_\lambda(s) = \bQ(s\theta) \left[\lim_{\lambda \to 0} \sum_{i\in \bbZ - \{0\}} \bK((\bQ(i\lambda\theta) - \bI)/\lambda \be_1 + i\delta \be_3) \right] \bQ(-s\theta) .
	\end{equation*}
	To take the limit inside the summation, we show that the sum is absolutely convergent for all $\lambda>0$ as follows:
	\begin{equation}
	    \begin{split}
    		a_\lambda &:= \sum_{i\in \bbZ - \{0\}} \left\vert \bK\left((\bQ(i\lambda\theta) - \bI)/\lambda \be_1 + i\delta \be_3\right) \right\vert \leq \sum_{i\in \bbZ - \{0\}} \dfrac{c}{\left\vert(\bQ(i\lambda\theta) - \bI)/\lambda \be_1 + i\delta \be_3)\right\vert^3}  \\
    		&= \sum_{i\in \bbZ - \{0\}} \dfrac{c}{(4\sin^2(i\lambda \theta/2)/\lambda^2 + i^2\delta^2)^{3/2}} \leq \sum_{i\in \bbZ - \{0\}} \dfrac{c}{|i|^3} < \infty , \qquad \forall \lambda >0 .
    	\end{split}
    \end{equation}
	Now, we can write
	\begin{equation*}
		\bH_0(s) =  \bQ(s\theta) \left[\sum_{i\in \bbZ - \{0\}} \lim_{\lambda \to 0}  \bK\left( \dfrac{ \bQ(i\lambda\theta) - \bI }{i\lambda\theta} (i\theta \be_1) + i\delta \be_3 \right) \right] \bQ(-s\theta) .
	\end{equation*}
	Note that for a fixed $i\in \bbZ$
	\begin{equation*}
		\lim_{\lambda \to 0} \dfrac{ (\bQ(i\lambda\theta)-\bI)}{i\lambda\theta} i\theta \be_1 + i\delta \be_3 = \lim_{h = i\lambda \theta \to 0} \dfrac{ (\bQ(h) -\bI) }{h} i\theta \be_1 + i\delta \be_3 = i \theta \bQ'(0) \be_1 + i\delta \be_3 ,
	\end{equation*}
	where $\bQ'(0) = \dmnosp/\dmnosp x \bQ(x)|_{x=0}$. Now, using the equation above, and the fact that $\bK(\bx)$ is smooth away from $\bx = \bzero$ (which is ensured in the summation), we get
	\begin{equation*}
		\bH_0(s) = \bQ(s\theta) \left[\sum_{i\in \bbZ - \{0\}} \bK\left( i \theta \bQ'(0) \be_1 + i\delta \be_3\right) \right] \bQ(-s\theta) = \bQ(s\theta) \bH_0(0)\bQ(-s\theta) .
	\end{equation*}
	Combining this with \eqref{eq:helixTLambdaLimitInter}, we get
	\begin{equation*}
		\lim_{\lambda \to 0} T_{\lambda} (\bff)(s) = \bH_0(s) \bff(s) =   \bQ(s\theta) \bH_0(0)\bQ(-s\theta) \bff(s) .
	\end{equation*}
	Next, we simplify $\bH_0(s)$. Using $\bQ\bK(\bx)\bQ^T = \bK(\bQ\bx)$ and $\bQ(s\theta)\bQ'(0) = \bQ'(s\theta)$, we can show
	\begin{equation}\label{eq:H0HelixIntermid}
		\bH_0(s) = \sum_{i\in \bbZ - \{0\}} \bK(i\theta \bQ'(s\theta) \be_1 + i\delta \be_3) = \sum_{i \in \bbZ - \{0\}} \bK\left(i\; |\bt(s)|\; \hat{\bt}(s)\right) ,
	\end{equation}
	where $\bt(s) = \theta \bQ'(s\theta) \be_1 + \delta \be_3$ is the tangent vector, and $\hat{\bt}(s) = \bt(s) / |\bt(s)|$ with $|\bt(s)| = \sqrt{\theta^2 + \delta^2}$. 
	In \eqref{eq:H0HelixIntermid}, by noting the definition of the dipole field kernel $\bK$, it is easy to show that 
	\begin{equation*}
	    \bH_0(s) = -h_0 \left[\bI - 3 \hat{\bt}(s) \dyad \hat{\bt}(s) \right] = -h_0 \left[\bP_{\perp}\bff(s) - 2\bP_{||}(s) \right],
	\end{equation*}
	with $h_0$ defined as
	\begin{equation}\label{eq:helixScalarh0}
		h_0 = \sum_{i \in \bbZ - \{0\}} \frac{1}{4\pi |i|^3 (\theta^2 + \delta^2)^{3/2}}
	\end{equation}
	and projection tensors $\bP_{||}(s) = \hat{\bt}(s)\dyad\hat{\bt}(s)$ and $\bP_{\perp}(s) = \bI - \bP_{||}(s)$. This finishes the proof of \sref{Proposition}{prop:helixTLambda}.

	\subsection{Nanofilm with Uniform Bending}\label{ss:filmProof}
	
	In this section, we prove \sref{Proposition}{prop:filmTLambda}. The outline of the proof is similar to the case of the helix in \autoref{ss:helixProof}.
	
	\begin{lemma}\label{lem:filmEstimates}
		
		\begin{itemize}
			\item[1.] Suppose $\bs, \bs' \in \calS_{1,\lambda} = (-\thetab/\lambda, \thetab/\lambda) \times \bbR$ such that $\ba, \bb \in \calL_{1,\lambda} = (-\thetab/\lambda, \thetab/\lambda)\cap \theta_l\bbZ \times \bbZ$ with $\bs \in U_1(\ba) =  [a_1, a_1 + \theta_l)\times [a_2, a_2+1), \bs' \in U_1(\bb)$. When $|\bs - \bs'| \geq \min\{\theta_l, 1\}$, we have $\ba \neq \bb$ and
			\begin{equation}\label{eq:filmMapEstimate1}
				\frac{|\bs - \bs'|}{|\ba - \bb|} < 1 + \frac{\theta_l + 1}{\min \{\theta_l, 1\}} =: c_L .
			\end{equation}
			
			\item[2.] For any $\ba, \bb\in \calL_{1,\lambda}$, we have
			\begin{equation}\label{eq:filmMapEstimate2}
				c_{A} |\ba - \bb| \leq |\bxbar(\ba) - \bxbar(\bb) + \bA_\lambda(\ba, \bb)\be_1| ,
			\end{equation}
			where $\bxbar$ is given by \eqref{eq:filmBxBar} and $\bA_\lambda(\ba, \bb)$ is defined as
			\begin{equation*}
				\bA_\lambda(\ba, \bb) = \frac{\calR}{\lambda} \left[\bQ(\lambda a_1 ) - \bQ(\lambda b_1) - \lambda \bQ(a_1) + \lambda \bQ(b_1) \right].
			\end{equation*}
			Here $c_A = \min\{\delta, \calR \sqrt{1 - \thetab^2/3}\}$ is the constant independent of $\lambda$; recall that $\delta$ is the parameter in the map $\bxbar$, see \eqref{eq:filmBxBar}. Note that $c_A > 0$ for $0 < \thetab < \pi/2$. 
		\end{itemize}
		
	\end{lemma}
	
	\begin{proof}
		To show \eqref{eq:filmMapEstimate1}, we proceed as follows. For $\bs, \bs'\in \calS_{1,\lambda}$ and corresponding $\ba,\bb \in \calL_{1,\lambda}$, there exists $\Delta \bs, \Delta\bs'$ such that $\bs = \ba + \Delta \bs, \bs' = \bb + \Delta \bb$ with $0\leq \Delta s_1, \Delta s'_1 < \theta_l$, $0\leq \Delta s_2, \Delta s'_2 < 1$. We have the bound
		\begin{equation}
			\frac{|\bs - \bs'|}{|\ba - \bb|} \leq 1 + \frac{|\Delta s_1 - \Delta s'_1| + |\Delta s_2 - \Delta s'_2|}{|\ba - \bb|} < 1 + \frac{\theta_l + 1}{|\ba - \bb|} \leq 1  + \frac{\theta_l + 1}{\min \{\theta_l, 1\}},
		\end{equation}
		where in the last step we used the fact that any $\ba,\bb \in \calL_{1,\lambda}$, satisfying $\ba \neq \bb$, are at least $\min\{\theta_l,1\}$ distance apart. 
		
		We next show \eqref{eq:filmMapEstimate2}. Using
		\begin{equation*}
			\bxbar(\ba) - \bxbar(\bb) + \bA_\lambda(\ba,\bb)\be_1 = \frac{\calR}{\lambda} (\bQ(\lambda a_1) - \bQ(\lambda b_1)) \be_1 + \delta (a_2 - b_2) \be_3 
		\end{equation*}
		and 
		\begin{equation*}
			|(\bQ(\theta_1) - \bQ(\theta_2))\be_1|^2 = (\cos\theta_1 - \cos\theta_2)^2 + (\sin\theta_1 - \sin\theta_2)^2 = 2(1 - \cos (\theta_1 - \theta_2)) ,
		\end{equation*}
		we have that 
		\begin{equation}\label{eq:filmBdd1}
			|\bxbar(\ba) - \bxbar(\bb) + \bA_\lambda(\ba,\bb)\be_1|^2 = \delta^2 |a_2 - b_2|^2 + \frac{2\calR^2}{\lambda^2} (1 - \cos(\lambda a_1 - \lambda b_1 )).
		\end{equation}
		Let $r = a_1 - b_1$. Then, using a Taylor expansion and the mean value theorem, there exists $\xi$ such that
		\begin{equation*}
			1 - \cos (\lambda r) = \frac{1}{2}\lambda^2  r^2 - \frac{1}{24} \lambda^4 r^4 \cos(\xi) . 
		\end{equation*}
		Since $-1 \leq \cos(\xi) \leq 1$, it follows
		\begin{equation*}
			1 - \cos (\lambda r) \geq  \frac{1}{2}\lambda^2 r^2 - \frac{1}{24} \lambda^4 r^4.
		\end{equation*}
		Substituting the relation above in \eqref{eq:filmBdd1}, we get
		\begin{equation*}
			|\bxbar(\ba) - \bxbar(\bb) + \bA_\lambda(\ba,\bb)\be_1|^2 \geq \delta^2 |a_2 - b_2|^2 + \calR^2 r^2 \left( 1- \frac{1}{12} \lambda^2 r^2\right).
		\end{equation*}
		Since $\ba,\bb \in \calL_{1,\lambda}$, we have $-2\thetab < \lambda r < 2\thetab$, and 
		\begin{equation*}
			1- \frac{1}{12} \lambda^2 r^2 \geq 1 - \frac{1}{12}\thetab^2 4 = 1 - \frac{\thetab^2}{3} .
		\end{equation*}
		Using the  two equations above and defining the constant $c_A$ as in \sref{Lemma}{lem:filmEstimates}(2), \eqref{eq:filmMapEstimate2} can be easily shown.
	\end{proof}
	
	\subsubsection{Boundedness}
	Let $S_\lambda \colon L^2(\calS,\bbR^3) \to L^2(\calS_{1,\lambda}, \bbR^3)$ be a map such that, for any $\bff \in L^2(\calS,\bbR^3)$,
	\begin{equation*}
		(S_\lambda\bff)(\bs) = \lambda \bff(\lambda \bs), \qquad \forall \bs \in \calS_{1,\lambda}.
	\end{equation*}
	It is easy to see that $S_\lambda$ is an isometry. The inverse of $S_\lambda$ is given by
	\begin{equation*}
		(S^{-1}_\lambda\bff)(\bs) = \lambda^{-1} \bff(\bs/\lambda), \qquad \forall \bs \in \calS .
	\end{equation*}
	Following the similar steps in \autoref{sss:helixProofBdd}, we can show that 
	\begin{equation*}
		||T_\lambda||_{\calL(L^2, L^2)} = ||T_{1,\lambda}||_{\calL(L^2, L^2)}.
	\end{equation*}
	Thus, to show that $T_\lambda$ is a bounded map, it is sufficient to show that $T_{1,\lambda}$ is bounded. Towards that goal, we first establish that
	\begin{equation}\label{eq:filmBddK1Lambda}
		|\bK_{1,\lambda}(\bs, \bs')| \leq \frac{C_1}{C_2 + |\bs - \bs'|^3}, \qquad \forall \bs', \bs \in \calS_{1,\lambda},
	\end{equation}
	where $C_1,C_2$ are constants that may depend on the parameters $R, \theta, \delta$ defining the surface $\calS$, but are independent of $\lambda$. 
	
	To show \eqref{eq:filmBddK1Lambda}, we recall that $\thetab$ is the fixed angular extent of the film and satisfies the bound $0 < \thetab < \pi/2$ (in fact we restrict it such that $0 < \thetab < \pi/4$). Let $\bs,\bs'\in \calS_{1,\lambda}$ be any two generic points, and let $\ba,\bb\in \calL_{1,\lambda}$ be such that $\bs \in U_1(\ba)$ and $\bs' \in U_1(\bb)$. We refer to \autoref{sss:filmDipoleSystem} and \autoref{sss:energyFilm} for the notation appearing in this section.
	
	First, consider $\bs, \bs'$ such that $|\bs - \bs'| \geq \min \{\theta_l, 1\}$. For this case, we have $\ba\neq \bb$. Noting that $
		|\bI - 3(\bx/|\bx|) \dyad (\bx/|\bx|)| = \sqrt{6}$, $\forall \bx \neq \bzero$, 
	we have
	\begin{equation}
	    \begin{split}\label{eq:filmBdd2}
    		|\bK_{1,\lambda}(\bs, \bs')| &\leq \frac{\sqrt{6}}{4\pi|\bxbar(\ba) - \bxbar(\bb) + \bA_\lambda(\ba,\bb)\be_1|^3}  \\
    		&= \frac{\sqrt{6}}{4\pi|\bs - \bs'|^3} \frac{|\bs - \bs'|^3}{|\ba - \bb|^3}\frac{|\ba - \bb|^3}{|\bxbar(\ba) - \bxbar(\bb) + \bA_\lambda(\ba,\bb)\be_1|^3}  \\
    		&\leq \frac{\sqrt{6}}{4\pi|\bs - \bs'|^3} c_L^3 \frac{1}{c_A^3},
	    \end{split}
    \end{equation}
	where we have used the bounds \eqref{eq:filmMapEstimate1} and \eqref{eq:filmMapEstimate2}. 
	
	Next,  we consider the case when $|\bs - \bs'| < \min\{\theta_l, 1\}$. This can be further divided in two cases:
	\begin{itemize}
		\item \textit{Case 1}: $\ba = \bb$ which implies $|\bK_{1,\lambda}(\bs',\bs)| = 0$.
		\item \textit{Case 2}: $\ba \neq \bb$. For this case, we have
		\begin{equation}\label{eq:filmBdd3}
			|\bK_{1,\lambda}(\bs, \bs')| \leq \frac{\sqrt{6}}{4\pi|\bxbar(\ba) - \bxbar(\bb) + \bA_\lambda(\ba,\bb)\be_1|^3} \leq \frac{\sqrt{6}}{4\pi c_A^3 |\ba - \bb|^3}.
		\end{equation}
		Note that when $\ba \neq \bb$, we can have either $a_1 = b_1, a_2 = b_2 \pm 1$; or $a_1 = b_1 \pm \theta_l,  a_2 = b_2$; or $a_1 = b_1 \pm \theta_l, a_2 = b_2 \pm 1$. For all of these cases, the denominator in \eqref{eq:filmBdd3} is bounded from below because $|\ba - \bb| \geq \min \{\theta_l, 1\}$. 
		Thus, we have
		\begin{equation}\label{eq:filmBdd4}
			|\bK_{1,\lambda}(\bs, \bs')| \leq \frac{\sqrt{6}}{4\pi c_A^3 (\min\{\theta_l, 1\})^3} .
		\end{equation}
	\end{itemize}
	In summary \eqref{eq:filmBdd4} holds for any $\bs, \bs'$ such that $|\bs - \bs'| < \min \{\theta_l, 1\}$.
	
	Combining the bound for the case $|\bs - \bs'| < \min\{\theta_l, 1\}$ with the bound for the case $|\bs - \bs'| \geq \min\{\theta_l, 1\}$, we can write
	\begin{equation}\label{eq:filmK1LambdaBddShow}
		|\bK_{1,\lambda}(\bs, \bs')|  \leq \frac{C_1}{C_2 + |\bs - \bs'|^3},
	\end{equation}
	where we have renamed the constants for convenience. This completes the proof of \eqref{eq:filmBddK1Lambda}. 
	
	Next, we show $T_{1,\lambda}$ is a bounded map on $L^2(\calS_{1,\lambda}, \bbR^3)$. Since $\bK_{1,\lambda}$ satisfies \eqref{eq:filmBddK1Lambda}, it can be shown that 
	\begin{equation}\label{eq:filmK1lamIntegBd}
	    \frac{\calR}{\theta_l}\int_{\calS_{1,\lambda}} |\bK_{1,\lambda}(\bs', \bs)| \dm \bs' \leq C_3, \quad \frac{\calR}{\theta_l}\int_{\calS_{1,\lambda}} |\bK_{1,\lambda}(\bs', \bs)| \dm \bs \leq C_3,
	\end{equation}
	for some fixed $C_3 < \infty$ independent of $\lambda$. Following the steps in obtaining inequality \eqref{eq:helixT1lambdaBddPf}, it is easy to obtain, for $\bff \in C^\infty_0(\calS_{1,\lambda}, \bbR^3)$,
	\begin{equation}
	    ||T_{1,\lambda} \bff||_{L^2(\calS_{1,\lambda}, \bbR^3)}^2 \leq C_3 \underbrace{\int_{\calS_{1,\lambda}} \left[\frac{\calR}{\theta_l} \int_{\calS_{1,\lambda}} |\bK_{1,\lambda}(\bt', \bs)|\,|\bff(\bt')|^2 \dm \bt' \right] \dm \bs}_{=: I} . \label{eq:filmT1lambdaBddPf}
	\end{equation}
    Let $\rho > 0$, and let $I_\rho$ is defined as
    \begin{equation}
        I_\rho := \int_{B_2(\bzero, \rho)\cap \calS_{1,\lambda}} \left[\frac{\calR}{\theta_l}\int_{\calS_{1,\lambda}} |\bK_{1,\lambda}(\bt', \bs)|\,|\bff(\bt')|^2 \dm \bt' \right] \dm \bs ,
    \end{equation}
    where $B_2(\bzero, \rho) = \{(s_1, s_2) \in \R^2: \sqrt{s_1^2 + s_2^2} \leq \rho\}$ is the two-dimensional ball of radius $\rho$ centered at $\bzero$. Based on the arguments in the last paragraph of \autoref{sss:helixProofBdd}, we find that $\lim_{\rho \to \infty} I_\rho$ exist and it is given by $\lim_{\rho \to \infty} I_\rho = I$, where $I$ is defined in \eqref{eq:filmT1lambdaBddPf}, and that the limit $I$ satisfies the bound $I \leq C_3 ||\bff||_{L^2(\calS_{1,\lambda}, \bbR^3)}^2$. Combining this with \eqref{eq:filmT1lambdaBddPf}, we have shown that, for $\bff \in C^\infty_0(\calS_{1,\lambda}, \bbR^3)$,
    \begin{equation*}
        ||T_{1,\lambda} \bff||_{L^2(\calS_{1,\lambda}, \bbR^3)} \leq C_3 ||\bff||_{L^2(\calS_{1,\lambda}, \bbR^3)} .
    \end{equation*}
    Arguing as in the case of the helix, the map $T_{1,\lambda}$
    is a bounded linear map on $L^2(\calS_{1,\lambda}, \bbR^3)$.

	\subsubsection{Limit of the Map $T_\lambda$}\label{ss:limit T nanofilm}
	
	Let $\bff\in C^{\infty}_0(\calS, \bbR^3)$. We write $T_\lambda \bff$ as follows
	\begin{equation}
	    \begin{split}\label{eq:filmTLambdaDecomp}
    		(T_\lambda \bff)(\bs) &= \frac{\calR\lambda}{\theta_l} \int_{\calS} \bK_\lambda(\bs', \bs) \bff(\bs') \dm \bs' \\
    		&= \underbrace{\left[\frac{\calR\lambda}{\theta_l} \int_{\calS} \bK_\lambda(\bs', \bs) \dm \bs' \right]}_{=: \bH_\lambda(\bs)} \bff(\bs) + \frac{\calR\lambda}{\theta_l} \int_{\calS} \bK_\lambda(\bs', \bs) (\bff(\bs') - \bff(\bs)) \dm \bs' .
    	\end{split}
    \end{equation}
	We next show that the second term in \eqref{eq:filmTLambdaDecomp} is zero in the limit $\lambda \to 0$. 
	Fix $\bar{R}>0$, 
	then we have
	\begin{equation}
	\label{eq:filmIonetwodef}
	\begin{split}
	    I & 
	    := \levert \lambda \int_{\calS} \bK_\lambda(\bs', \bs) (\bff(\bs') - \bff(\bs)) \dm \bs' \rivert 
	    \\ & 
		\leq 
		\levert \lambda \int_{\calS} \chi_{|\bs - \bs'|\geq \bar{R}\lambda}(\bs')  \bK_\lambda(\bs', \bs) (\bff(\bs') - \bff(\bs)) \dm \bs' \rivert
		\\ & 
		\quad + 
		\levert \lambda \int_{\calS} \chi_{|\bs - \bs'|\leq \bar{R}\lambda}(\bs')  \bK_\lambda(\bs', \bs) (\bff(\bs') - \bff(\bs)) \dm \bs' \rivert
		\\ &
		\leq 
		\lambda \int_{\calS} \chi_{|\bs - \bs'|\geq \bar{R}\lambda}(\bs')  |\bK_\lambda(\bs', \bs)|\, |\bff(\bs') - \bff(\bs)| \dm \bs'
		\\ & 
		\quad + 
		\lambda \int_{\calS} \chi_{|\bs - \bs'|\leq \bar{R}\lambda}(\bs')  |\bK_\lambda(\bs', \bs)| \, |\bff(\bs') - \bff(\bs)| \dm \bs'
		\\ &
		\leq \left(2 \sup_{\bs'} |\bff(\bs')|\right) \underbrace{\lambda  \int_{\calS} \chi_{|\bs - \bs'|\geq \bar{R}\lambda}(\bs')  |\bK_\lambda(\bs', \bs)| \dm \bs'}_{=: I_1}
		\\ & 
		\quad +
		\left(\bar{R}  \sup_{\bs'} |\nabla \bff(\bs')|\right) \underbrace{\lambda^2 \int_{\calS} \chi_{|\bs - \bs'|\leq \bar{R}\lambda}(\bs')  |\bK_\lambda(\bs', \bs)| \dm \bs'}_{=: I_2}, 
	\end{split}
    \end{equation}
	where, in the last step, we have used that $|\bff(\bs') - \bff(\bs)| \leq 2 \sup_{\bt} |\bff(\bt)| $, and, for $\bs' \in \{\bt: |\bt - \bs| \leq \bar{R}\lambda\}$, we have $|\bff(\bs') - \bff(\bs)| \leq \bar{R} \lambda \sup_{\bt} |\nabla \bff(\bt)|$. 
	To bound $I_1$ and $I_2$, first observe that, due to $\bK_\lambda(\bs', \bs) = \bK_{1,\lambda}(\bs'/\lambda, \bs/\lambda)/\lambda^3$ and the bound on $\bK_{1,\lambda}$ in \eqref{eq:filmK1LambdaBddShow}, there holds
	\begin{equation}\label{eq:filmKlamBdd}
	    |\bK_\lambda(\bs', \bs)| = \frac{1}{\lambda^3} |\bK_{1,\lambda}(\bs'/\lambda, \bs/\lambda)| \leq \frac{1}{\lambda^3} \frac{C_1}{C_2 + |\bs' - \bs|^3/\lambda^3},
	\end{equation}
	and as a direct consequence, we have
	\begin{equation}
	\begin{split}
	    I_1 &\leq \lambda \frac{1}{\lambda^3} \int_{\calS} \chi_{|\bs - \bs'|\geq \bar{R}\lambda}(\bs') \frac{C_1}{C_2 + |\bs' - \bs|^3/\lambda^3} \dm \bs' \\
	    &\leq \frac{1}{\lambda^2} \int_{|\bs - \bs'|\geq \bar{R}\lambda} \frac{C_1}{C_2 + |\bs' - \bs|^3/\lambda^3} \dm \bs' \\
	    &= \int_{|\bt| \geq \bar{R}} \frac{C_1}{C_2 + |\bt|^3} \dm \bt, 
	\end{split}
	\label{eq:filmIoneBddForT1lamBd}
	\end{equation}
	where, in the second step, the domain of integration was enlarged; and, in the final step, the change of variables $\bt = (\bs' - \bs)/\lambda$ (so that $\lambda^2 \dm \bt = \dm \bs'$) was introduced. Next, using the definition of $I_2$ in \eqref{eq:filmIonetwodef} and \eqref{eq:filmKlamBdd}, we get
	\begin{equation}
	\begin{split}
	    I_2 &\leq \lambda^2 \frac{1}{\lambda^3} \int_{\calS} \chi_{|\bs - \bs'|\leq \bar{R}\lambda}(\bs') \frac{C_1}{C_2 + |\bs' - \bs|^3/\lambda^3} \dm \bs' \\
	    &\leq \frac{1}{\lambda} \int_{|\bs - \bs'|\leq \bar{R}\lambda} \frac{C_1}{C_2 + |\bs' - \bs|^3/\lambda^3} \dm \bs' \\
	    &\leq \lambda \int_{|\bt| \leq \bar{R}} \frac{C_1}{C_2 + |\bt|^3} \dm \bt \\
	    &\leq \left[ \frac{C_1}{C_2} \pi \bar{R}^2\right] \lambda, 
	\end{split}
	\label{eq:filmItwoBddForT1lamBd}
	\end{equation}
	where, in the second step, the domain of integration was enlarged; in the third step, the change of variables $\bt = (\bs' - \bs)/\lambda$ was introduced; and, finally, in the last step, $C_1/(C_2 + |\bt|^3) \leq C_1/C_2$ was used. Combining the results of \eqref{eq:filmIoneBddForT1lamBd} and \eqref{eq:filmItwoBddForT1lamBd} with \eqref{eq:filmIonetwodef}, we have shown that
	\begin{equation}
	\begin{split}
	    I &\leq \left(2 \sup_{\bs'} |\bff(\bs')|\right)\, I_1 + \left(\bar{R}  \sup_{\bs'} |\nabla \bff(\bs')|\right)\, I_2  \\
	    &\leq \left(2 \sup_{\bs'} |\bff(\bs')|\right)\, \int_{|\bt| \geq \bar{R}} \frac{C_1}{C_2 + |\bt|^3} \dm \bt + \left(\bar{R}  \sup_{\bs'} |\nabla \bff(\bs')|\right)\, \left[ \frac{C_1}{C_2} \pi \bar{R}^2\right] \lambda,
	\end{split}
	\end{equation}
	for any $\bar{R} > 0$. Arguing as in the case of helix -- see the discussion associated with \eqref{eq:helixUpperBdd4} -- we have
	\begin{equation*}
	    \lim_{\lambda \to 0}I \leq \left(2 \sup_{\bs'} |\bff(\bs')|\right)\, \int_{|\bt| \geq \bar{R}} \frac{C_1}{C_2 + |\bt|^3} \dm \bt ,
	\end{equation*}
	for any $\bar{R} > 0$. We can take the limit $\bar{R} \to \infty$ of both sides above since the inequality is valid for any $\bar{R} > 0$, and the left side is independent of $\bar{R}$ and the limit of $\bar{R} \to \infty$ of the right side is well-defined and equal to $0$. We have therefore shown
	\begin{equation*}
	    \lim_{\lambda \to 0}I = \lim_{\lambda \to 0} \levert \lambda \int_{\calS} \bK_\lambda(\bs', \bs) (\bff(\bs') - \bff(\bs)) \dm \bs' \rivert \leq \lim_{\bar{R} \to \infty} \left[ \left(2 \sup_{\bs'} |\bff(\bs')|\right)\, \int_{|\bt| \geq \bar{R}} \frac{C_1}{C_2 + |\bt|^3} \dm \bt  \right] = 0.
	\end{equation*}

	Thus, from \eqref{eq:filmTLambdaDecomp}, we have
	\begin{equation}\label{eq:filmTLambdaDecomp1}
		\lim_{\lambda \to 0}(T_\lambda \bff)(\bs) = \left[ \lim_{\lambda \to 0} \bH_\lambda(\bs) \right] \bff(\bs).
	\end{equation}
	
	\paragraph{Limit of $\bH_\lambda$}
	
	Consider a typical $\bs \in \calS$ such that $\bs \in U_\lambda(\ba)$ where $\ba \in \calL_\lambda$. 
	Recall that $\calL_\lambda$ is the lattice for $\lambda >0$ and $U_\lambda(\ba) = [a_1, a_1 + \theta_l \lambda) \times [a_2, a_2 + \lambda)$ is the lattice cell. In the definition of $\bH_\lambda$, we substitute $\bK_\lambda$, to get
	\begin{equation}\label{eq:filmHLambdaLimit1}
		\bH_\lambda(\bs) = \frac{\calR\lambda}{\theta_l} \sum_{\bu \in \calL_\lambda, \bu \neq \ba} \bK(\bxbar(\bu) - \bxbar(\ba)) \int_{U_\lambda(\bu)} \dm \bs' =  \calR \lambda^3 \sum_{\bu \in \calL_\lambda, \bu \neq \ba} \bK(\bxbar(\bu) - \bxbar(\ba)).
	\end{equation}
	Substituting the definition of transformation $\bxbar$ in \eqref{eq:filmBxBar}, we can show for $\ba,\bu \in \calL_\lambda$ that
	\begin{equation*}
		\bxbar(\bu) - \bxbar(\ba) = \bQ(a_1\theta) \left[ \calR( \bQ(u_1 - a_1) - \bI) \be_1 + (u_2 - a_2)\delta \be_3 \right].
	\end{equation*}
	Using the identities $\bK(\bQ(t)\bx) = \bQ(t) \bK(\bx) \bQ^T(t)$ and $\bK(\lambda\bx) = \bK(\bx)/\lambda^3$, from \eqref{eq:filmHLambdaLimit1}, we have
	\begin{equation}\label{eq:filmHLambdaLimit11}
		\bH_\lambda(\bs) =  \calR\bQ(a_1) \underbrace{\left[ \sum_{\bu \in \calL_\lambda, \bu \neq \ba} \bK\left(\calR( \bQ(u_1 - a_1) - \bI)/\lambda \be_1 + (u_2 - a_2)\delta/\lambda \be_3\right) \right]}_{=: \bar{\bH}_\lambda(\bs)} \bQ(-a_1) .
	\end{equation} 
	We analyze $\bar{\bH}_\lambda$ as follows. First, we expand the sum $\bu \in \calL_\lambda$ 
	\begin{equation}
	    \begin{split}\label{eq:filmHLambdaLimit2}
    		\bar{\bH}_\lambda(\bs) &= \sum_{u_2 \in \lambda \bbZ} \left[ \sum_{\substack{u_1 \in \lambda\theta_l \bbZ \cap (-\thetab,\thetab),\\
    				(u_1, u_2) \neq \ba}} \bK\left(\calR \frac{\bQ(u_1 - a_1) - \bI}{\lambda} \be_1 + \delta \frac{u_2 - a_2}{\lambda} \be_3 \right) \right] \\
    		&= \sum_{t_2' \in \bbZ} \left[ \sum_{\substack{u_1 \in \lambda\theta_l \bbZ \cap (-\thetab,\thetab),\\
    				(u_1, t_2') \neq (a_1, 0)}} \bK\left(\calR \frac{\bQ(u_1 - a_1) - \bI}{\lambda} \be_1 + \delta t_2' \be_3 \right) \right] ,
	    \end{split}
    \end{equation}
	where we introduced the new variable $t_2' = (u_2 - a_2)/\lambda$. Since $u_2,a_2 \in \lambda \bbZ$, we have $t_2' \in \bbZ$. Using a Taylor expansion and the mean value theorem, we have the identity
	\begin{equation}\label{eq:rotTransformTaylor}
		\bQ(u_1 - a_1) - \bI = \bQ'(\xi) (u_1 - a_1),
	\end{equation}
	where $\xi = \xi(u_1 - a_1) \in (-\thetab, \thetab)$ depends on $u_1 - a_1$. 
	Formula \eqref{eq:rotTransformTaylor} suggests that we decompose \eqref{eq:filmHLambdaLimit2} as:
	\begin{equation}
	    \begin{split}\label{eq:filmHLambdaLimit3}
    		\bar{\bH}_\lambda(\bs)
    		= & \underbrace{\sum_{t_2' \in \bbZ} \left[ \sum_{\substack{u_1 \in \lambda\theta_l \bbZ \cap (-\thetab,\thetab),\\
    				(u_1, t_2') \neq (a_1, 0)}} \bK\left(\calR \bQ'(0) \frac{u_1 - a_1}{\lambda} \be_1 + \delta t_2' \be_3 \right) \right]}_{=:\bar{\bH}_\lambda^{(1)}(\bs)} 
    		\\
    		& + \underbrace{ \sum_{t_2' \in \bbZ} \left[ \sum_{\substack{u_1 \in \lambda\theta_l \bbZ \cap (-\thetab,\thetab),\\
    				(u_1, t_2') \neq (a_1, 0)}} \left\{ \bK\left(\calR \frac{\bQ(u_1 - a_1) - \bI}{\lambda} \be_1 + \delta t_2' \be_3 \right) - \bK\left(\calR \bQ'(0) \frac{u_1 - a_1}{\lambda} \be_1 + \delta t_2' \be_3 \right) \right\}
    		\right] }_{=: \bar{\bH}_\lambda^{(2)}(\bs)} .
    	\end{split}
    \end{equation}

	\textbf{Step 1:} We show $\bar{\bH}_\lambda^{(2)}$ goes to zero in the limit $\lambda \to 0$. Let
	\begin{equation}\label{eq:filmX12}
		\bx_1 = \calR \frac{\bQ(u_1 - a_1) - \bI}{\lambda} \be_1, \quad \bx_2 = \calR \bQ'(0) \frac{u_1 - a_1}{\lambda} \be_1, \quad \bz = \delta t_2'\be_3.
	\end{equation}
	Consider a function $\by\colon [0,1] \to \bbR^3$ defined as
	\begin{equation}\label{eq:filmFnY}
		\by(r) = \bx_1  + r(\bx_2 - \bx_1).
	\end{equation}
	Note that, since $(u_1, t_2') \neq (a_1, 0)$ and $t_2' \in \bbZ$, we have
	\begin{equation}\label{eq:filmYZMinShow}
		|\by(r) + \bz| \geq \min \{\delta,  \min_{r\in [0,1], u_1 \in (\lambda\theta_l \bbZ - \{a_1\})\cap (-\thetab, \thetab)} |\by(r)| \}.
	\end{equation}
	We show that $\min_{r\in [0,1], u_1 \in (\lambda \theta_l \bbZ - \{a_1\})\cap (-\thetab, \thetab)} |\by(r)| > 0$ and the lower bound is independent of $\lambda$. For convenience, let $t = u_1 - a_1$. Since $u_1, a_1 \in \lambda\theta_l \bbZ \cap (-\thetab,\thetab)$, and $u_1 \neq a_1$, we have $t \in (\lambda\theta_l \bbZ - \{0\}) \cap (-2\thetab, 2\thetab)$. The hypothesis of \sref{Proposition}{prop:filmTLambda} restricts $\thetab$ such that
	\begin{equation}\label{eq:filmThetaRestrict}
		0 < \thetab < \pi/4 \quad \Rightarrow 0 < \cos (2\thetab) < 1.
	\end{equation}
	With $t = u_1 - a_1$, writing out the action of $\bQ(t)$ and $\bQ'(0)$ on $\be_1$, we get
	\begin{equation}
	    \begin{split}\label{eq:filmX12Expand}
    		\bx_1 &= \calR \frac{\bQ(t) - \bI}{\lambda} \be_1 =  \frac{\calR}{\lambda}[ (\cos(t) - 1) \be_1 + \sin(t) \be_2], \\
    		\bx_2 &= \calR \bQ'(0) \frac{t}{\lambda} \be_1 = \frac{\calR t}{\lambda} [-\sin(t)\be_1 + \cos(t) \be_2].
    	\end{split}
    \end{equation}
	Through elementary calculations, we can show
	\begin{equation}\label{eq:filmYSq}
		|\by(r)|^2 = |\bx_1 + r (\bx_2 - \bx_1)|^2 = \frac{\calR^2}{\lambda^2} \left[2(1-r)^2 (1- \cos(t)) + r^2 t^2 + 2r(1-r) t \sin(t) \right].
	\end{equation}
	Using a Taylor expansion and noting that $t\in (\lambda \theta_l \bbZ - \{0\}) \cap (-2\thetab, 2\thetab)$, there exists $\xi_1,\xi_2 \in (-2\thetab, 2\thetab)$ with $\xi_1 = \xi_1(t), \xi_2 = \xi_2(t)$ such that
	\begin{equation}
		1 - \cos(t) = \cos (\xi_1) t^2/2, \quad \sin(t) = t\cos(\xi_2).
	\end{equation}
	Thus
	\begin{equation}
	    \begin{split}\label{eq:filmYEstimate}
    		|\by(r)|^2 &= \frac{\calR^2}{\lambda^2} \left[2(1-r)^2 \cos(\xi_1) t^2/2 + r^2 t^2 + 2r(1-r) t \cos(\xi_2)t \right] \\
    		&= \frac{\calR^2 t^2}{\lambda^2} \left[ (1-r)^2 \cos(\xi_1) + r^2 + 2r(1-r) \cos(\xi_2) \right] \\
    		&\geq \frac{\calR^2 t^2}{\lambda^2} \left[ (1-r)^2 \min_{\xi \in (-2\thetab, 2\thetab)}\cos(\xi) + r^2 + 2r(1-r) \min_{\xi \in (-2\thetab, 2\thetab)}\cos(\xi) \right] \\
    		&= \frac{\calR^2 t^2}{\lambda^2} \left[ (1-r)^2 \cos(2\thetab) + r^2 + 2r(1-r) \cos(2\thetab) \right] \\
    		&\geq  \frac{\calR^2 t^2}{\lambda^2} \min_{r\in [0,1]} \left[ (1-r)^2 \cos(2\thetab) + r^2 + 2r(1-r) \cos(2\thetab) \right] \\
    		&= \frac{\calR^2 t^2}{\lambda^2} \cos(2\thetab),
    	\end{split}
    \end{equation}
	where we used the fact that $\min_{\xi\in (-2\thetab, 2\thetab)} \cos(\xi) = \cos(2\thetab)$ in the fourth line, and $\cos(2\thetab)$ is the minimum with respect to $r\in [0,1]$ of the function in the square bracket in the fifth line. Further, since $t\in (\lambda \theta_l \bbZ - \{0\}) \cap (-2\thetab, 2\thetab)$, we have
	\begin{equation}\label{eq:filmYMin}
		0 < C_y := \frac{\calR^2 \lambda^2 \theta_l^2}{\lambda^2}\cos(2\thetab) = (\calR \theta_l)^2 \cos(2\thetab) \leq |\by(r)|^2,
	\end{equation}
	for any $t\in (\lambda \theta_l \bbZ - \{0\}) \cap (-2\thetab, 2\thetab)$ and $r\in [0,1]$. The lower bound on $|\by(r)|$ is independent of $\lambda$ and $r$. Finally, combining \eqref{eq:filmYMin} with \eqref{eq:filmYZMinShow}, we get
	\begin{equation}\label{eq:filmYZMin}
		0 < C_{yz} := \min \{\delta, \calR \theta_l\sqrt{\cos(2\thetab)} \} \leq |\by(r) + \bz|.
	\end{equation}
	Proceeding further, we have, from the fundamental theorem of calculus,
	\begin{equation}
	    \begin{split}\label{eq:filmFundThmCalculus}
    		\bK(\bx_1 + \bz) - \bK(\bx_2 + \bz) &= \int_0^1 \frac{\dm}{\dm r} \bK(\by(r) + \bz) \dm r = \int_0^1 \nabla \bK(\by(r) + \bz) \frac{\dm }{\dm r} \by(r) \dm r \\
    		&= \int_0^1 \nabla \bK(\by(r) + \bz) (\bx_2 - \bx_1) \dm r.
    	\end{split}
    \end{equation}
	Note that because of \eqref{eq:filmYZMin}, $\nabla \bK(\by(r) + \bz)$ exists and is bounded. From the definition of $\bar{\bH}_\lambda^{(2)}$ in \eqref{eq:filmHLambdaLimit3}, a change of variable $t = u_1 - a_1$, the definition of $\bx_1, \bx_2,\bz$ in \eqref{eq:filmX12} and \eqref{eq:filmX12Expand}, and noting the identity \eqref{eq:filmFundThmCalculus}, we have 
	\begin{equation}
	    \begin{split}\label{eq:filmH2LambdaEstimate1}
    		|\bar{\bH}_\lambda^{(2)}(\bs)|
    		&\leq \sum_{t_2' \in \bbZ} \left[ \sum_{\substack{u_1 \in \lambda\theta_l \bbZ \cap (-\thetab,\thetab),\\
    				(u_1, t_2') \neq (a_1, 0)}} \left\vert \bK\left(\calR \frac{\bQ(u_1 - a_1) - \bI}{\lambda} \be_1 + \delta t_2' \be_3 \right) - \bK\left(\calR \bQ'(0) \frac{u_1 - a_1}{\lambda} \be_1 + \delta t_2' \be_3 \right) \right\vert \right]  \\
    		&\leq  \sum_{t_2' \in \bbZ} \left[ \sum_{t \in \lambda\theta_l \bbZ - \{0\} \cap (-2\thetab,2\thetab)} \left\vert \bK\left(\bx_1 + \bz \right) - \bK\left(\bx_2 + \bz \right) \right\vert
    		\right]  \\
    		&\leq \sum_{t_2' \in \bbZ} \left[ \sum_{t \in \lambda\theta_l \bbZ - \{0\} \cap (-2\thetab,2\thetab)} \int_0^1 \left\vert \nabla \bK(\by(r) + \bz) \right\vert \; |\bx_2 - \bx_1| \dm r \right]  \\
    		&\leq \sum_{t_2' \in \bbZ} \left[ \sum_{t \in \lambda\theta_l \bbZ - \{0\} \cap (-2\thetab,2\thetab)} \int_0^1 \frac{C}{|\by(r) + \bz|^4} |\bx_2 - \bx_1| \dm r \right]  \\
    		&= \sum_{t_2' \in \bbZ} \left[ \sum_{t \in \lambda\theta_l \bbZ - \{0\} \cap (-2\thetab,2\thetab)} \int_0^1 \frac{C}{(|\by(r)|^2 + |\bz|^2)^2} |\bx_2 - \bx_1| \dm r \right],
	    \end{split}
    \end{equation}
	where we used the bound on the gradient of $\bK$ with constant $C>0$ fixed. 
	
	Next, we get an upper bound on $|\bx_1 - \bx_2|$ in terms of $t$. From \eqref{eq:filmX12Expand}, we have
	$$
	\bx_2 - \bx_1 = \frac{\calR}{\lambda} \left[ t\bQ'(0) - \bQ(t) + \bI  \right]\be_1.
	$$
	By a Taylor expansion and the mean value theorem, we have $\bQ(t) = \bI + t\bQ'(0) + (t^2/2)\bQ''(\xi)$ where $\xi = \xi(t) \in (-2\thetab, 2\thetab)$ depends on $t$.
	Substituting this and using the bound $|\bQ''_{ij}(\xi)| \leq 1$, we obtain
	\begin{equation}\label{eq:filmBoundX12}
		|\bx_2 - \bx_1| = \frac{\calR}{\lambda} \frac{|t|^2}{2} |\bQ''(\xi)| \leq  \frac{\calR}{\lambda} \frac{|t|^2}{2}.
	\end{equation}
	Combining the equation above with \eqref{eq:filmH2LambdaEstimate1}, we get
	\begin{equation}
	    \begin{split}
    		|\bar{\bH}_\lambda^{(2)}(\bs)|
    		&
    		\leq \sum_{t_2' \in \bbZ} \left[ \sum_{t \in \lambda\theta_l \bbZ - \{0\} \cap (-2\thetab,2\thetab)} \int_0^1 \frac{C}{(|\by(r)|^2 + |\bz|^2)^2}  \frac{\calR}{\lambda} \frac{|t|^2}{2} \dm r \right]  \\
    		&= \sum_{t_2' \in \bbZ} \left[ \sum_{t' \in \theta_l \bbZ - \{0\} \cap (-2\thetab/\lambda,2\thetab/\lambda)} \int_0^1 \frac{C}{(|\by(r)|^2 + |\bz|^2)^2}  \frac{\calR}{\lambda} \frac{\lambda^2|t'|^2}{2} \dm r \right]  \\
    		&\leq \lambda \left\{ \sum_{t_2' \in \bbZ} \left[ \sum_{t' \in \theta_l \bbZ - \{0\} \cap (-2\thetab/\lambda,2\thetab/\lambda)} \int_0^1 \frac{C}{(|\by(r)|^2 + |\bz|^2)^2}  \frac{\calR|t'|^2}{2} \dm r \right] \right\},
    	\end{split}
    \end{equation}
	where in the third line we introduced the variable $t' = t/\lambda$. We only have to show that the term inside the braces is bounded as $\lambda \to 0$ to conclude that $|\bar{\bH}_\lambda^{(2)}(\bs)| \to 0$ as $\lambda \to 0$. First, note from \eqref{eq:filmYEstimate}, we have
	\begin{equation}
		|\by(r)|^2 \geq \frac{\calR^2}{\lambda^2} |t|^2 \cos(2\thetab) = \calR^2 |t'|^2 \cos(2\thetab).
	\end{equation}
	Therefore,
	\begin{equation}
		\frac{C}{(|\by(r)|^2 + |\bz|^2)^2} \leq \frac{C}{(  \calR^2 |t'|^2 \cos(2\thetab) + |\bz|^2)^2}.
	\end{equation}
	Thus
	\begin{equation}\label{eq:filmEst0}
		|\bar{\bH}_\lambda^{(2)}(\bs)| 
		\leq \lambda \left\{ \sum_{t_2' \in \bbZ} \left[ \sum_{t' \in \theta_l \bbZ - \{0\} \cap (-2\thetab/\lambda,2\thetab/\lambda)} \int_0^1 \frac{C}{(\calR^2 |t'|^2 \cos(2\thetab) + |\bz|^2)^2}  \frac{\calR|t'|^2}{2} \dm r \right] \right\}.
	\end{equation}
	Note that the integrand is independent of $r$. Further, the numerator has $|t'|^2$ whereas the denominator has $(|t'|^2 c + |\bz|^2)^2$, therefore, the sum inside the braces is absolutely convergent and finite. Hence, due to the factor $\lambda$, we have shown $\lim_{\lambda\to 0} |\bar{\bH}_\lambda^{(2)}(\bs)| = 0$.
	
	This completes Step 1. We next study $\bar{\bH}_\lambda^{(1)}$.
	
	\textbf{Step 2:} We have from \eqref{eq:filmHLambdaLimit3}
	\begin{equation}
	    \begin{split}\label{eq:filmH1LambdaEstimate1}
    		\bar{\bH}_\lambda^{(1)}(\bs) &= \sum_{t_2' \in \bbZ} \left[ \sum_{\substack{u_1 \in \lambda\theta_l \bbZ \cap (-\thetab,\thetab),\\
    				(u_1, t_2') \neq (a_1, 0)}} \bK\left(\calR \bQ'(0) \frac{u_1 - a_1}{\lambda} \be_1 + \delta t_2' \be_3 \right) \right] \\
    		&= \underbrace{
    		    \sum_{t_2' \in \bbZ} \left[ 
    		        \sum_{\substack{u_1 \in \lambda\theta_l \bbZ,\\ (u_1, t_2') \neq (a_1, 0)}} \bK\left(\calR \bQ'(0) \frac{u_1 - a_1}{\lambda} \be_1 + \delta t_2' \be_3 \right) 
                \right]
                }_{=:I_1}
    	    \\
    		& \quad - \underbrace{
    		    \sum_{t_2' \in \bbZ} \left[ \sum_{\substack{u_1 \in [\lambda \theta_l \bbZ] - [\lambda\theta_l \bbZ \cap (-\thetab,\thetab)],\\
    				(u_1, t_2') \neq (a_1, 0)}} \bK\left(\calR \bQ'(0) \frac{u_1 - a_1}{\lambda} \be_1 + \delta t_2' \be_3 \right) \right]
    			}_{=:I_2},
    	\end{split}
    \end{equation}
	where we have used the notation $[\lambda \theta_l \bbZ] - [\lambda\theta_l \bbZ \cap (-\thetab,\thetab)]$ to denote the set $\{t \in\lambda \theta_l \bbZ;\, t \notin \lambda\theta_l \bbZ \cap (-\thetab,\thetab) \}$. Using the decay property of the dipole field kernel $\bK$, we can show that $|I_2| \to 0$ in the limit $\lambda \to 0$. Therefore, we have
	\begin{equation}
	    \begin{split}\label{eq:filmH1LambdaEstimate2}
    		\lim_{\lambda \to 0} \bar{\bH}_\lambda^{(1)}(\bs) 
    		&= \lim_{\lambda \to 0} \sum_{t_2' \in \bbZ} \left[ \sum_{\substack{u_1 \in \lambda\theta_l \bbZ,\\
    				(u_1, t_2') \neq (a_1, 0)}} \bK\left(\calR \bQ'(0) \frac{u_1 - a_1}{\lambda} \be_1 + \delta t_2' \be_3 \right) \right] \\
    		&= \sum_{t_2' \in \bbZ} \left[ \sum_{\substack{t_1' \in \theta_l \bbZ,\\
    				(t_1', t_2') \neq (0, 0)}} \bK\left(\calR \bQ'(0) t_1' \be_1 + \delta t_2' \be_3 \right) \right], \\
	\end{split}
    \end{equation}
	where we introduced the new variable $t_1' = (u_1 - a_1)/\lambda$. Since $u_1 \in \lambda \theta_1\bbZ$ and $a_1 \in \lambda\theta_l \bbZ\cap (-\thetab,\thetab)$, we have $t_1' \in \theta_l \bbZ$. This completes Step 2. Note that $\lim_{\lambda \to 0} \bar{\bH}_{\lambda}(\bs)$ is independent of $\bs \in \calS$. 
	
	Upon substituting the limit of $\bar{\bH}_\lambda^{(1)}$ and $\bar{\bH}_\lambda^{(2)}$ in \eqref{eq:filmHLambdaLimit3}, we have shown
	\begin{equation}\label{eq:filmHBarLambdaLimit}
		\lim_{\lambda\to 0} \bar{\bH}_\lambda(\bs) = \lim_{\lambda\to 0} \bar{\bH}_\lambda(\bzero) = \sum_{\substack{\bu = (u_1,u_2) \in \theta_l\bbZ \times \bbZ, \\ \bu \neq \bzero}} \bK\left(\calR \bQ'(0) u_1 \be_1 + \delta u_2 \be_3 \right) .
	\end{equation}
	Recall that $\bs \in \calS$ was fixed such that $\bs \in U_\lambda (\ba)$, which implies that $\ba \to \bs$ as $\lambda\to 0$. With this observation and \eqref{eq:filmHBarLambdaLimit}, we have from \eqref{eq:filmHLambdaLimit11},
	\begin{equation}\label{eq:filmH0Final}
		\bH_0(\bs) = \lim_{\lambda\to 0} \bH_\lambda(\bs) = \calR\bQ(s_1) \left[\sum_{\substack{\bu = (u_1,u_2) \in \theta_l\bbZ \times \bbZ, \\ \bu \neq \bzero}} \bK\left(\calR \bQ'(0) u_1 \be_1 + \delta u_2 \be_3 \right) \right] \bQ(-s_1).
	\end{equation}
	
	Next we simplify $\bH_0(\bs)$. Given the parametric map $\bxbar = \bxbar(\bs)$, the two tangent vectors at $\bs = (s_1, s_2)$ are 
	\begin{equation}\label{eq:filmTangentVectors}
		\bt_1(\bs) = \frac{\dmnosp \bxbar}{\dmnosp s_1} = \calR \bQ'(s_1) \be_1, \qquad \bt_2(\bs) = \frac{\dmnosp \bxbar}{\dmnosp s_2} = \delta \be_3.
	\end{equation}
	Using $\bQ\bK(\bx) \bQ^T = \bK(\bQ\bx)$ and $\bQ(r) \bQ'(0) = \bQ'(r)$, we write,
	\begin{equation*}
		\bH_0(\bs) = \sum_{\substack{\bu = (u_1,u_2) \in \theta_l\bbZ \times \bbZ, \\ \bu \neq \bzero}} \bK\left(u_1 \bt_1(\bs) + u_2 \bt_2(\bs)\right).
	\end{equation*}
	This completes the proof of \sref{Proposition}{prop:filmTLambda}.

 	\section{Summary of Results}\label{s:conclusion}
	
	We have shown rigorously that certain low-dimensional nanostructures do not have long-range dipole-dipole interaction in the continuum limit. 
	The energy density in the limit is entirely because of the Maxwell self-field. 
	In 1-d and 2-d lattices (in a 3-d ambient space), the dipole field kernel decay is sufficiently fast that long-range interactions do not contribute to the limit energy.

	While our calculations show that the energy is local in the continuum limit for 1-d and 2-d discrete systems, in agreement with dimension reduction approaches that reduce a 3-d continuum to a 1-d or 2-d continuum (e.g., \cite{gioia1997micromagnetics, chacouche2015ferromagnetic} and others), we note an interesting difference.
    	As shown in \cite{gioia1997micromagnetics} and other work following it, the component of the dipole moment along the normal direction to the film is the only contributor to the continuum electrostatic energy in the thin film limit.
    	Similarly, \cite{chacouche2015ferromagnetic} show that the component of the dipole moment in the plane normal to the wire is the only contributor to the continuum electrostatic energy in the thin wire limit.
	
	This is different from the limit energy in the discrete-to-continuum limit obtained in this work: for the case of a helical nanotube, the limiting energy density is given by (see \autoref{thm:helixELambda}) 
	$$
    	h_0 \int_{\bbR} \left[ |\bP_{\perp}\bff|^2 - 2 |\bP_{||}\bff|^2 \right] \dm s,
	$$
	where $h_0$ is a constant, and $\bP_{||} \bff$ and $\bP_{\perp} \bff$ are, respectively, the projections of the dipole moment field $\bff$ along the axis of the helix and in the plane normal to the axis of the helix.
	Therefore, unlike the thin wire limit using dimension reduction, the discrete-to-continuum energy has contributions from both the normal and tangential components of the dipole moment field. 
	For the case of a thin film with curvature, the limiting energy density is given by (see \autoref{thm:filmELambda})
	\begin{equation*}
	    -\frac{1}{2} \int_{\calS} \bff(\bs) \cdot \bH_0(\bs) \bff(\bs) \dm \bs,
	\end{equation*}
	where 
	\begin{equation*}
	    \bH_0(\bs) =  \calR \sum_{\substack{\bu = (u_1,u_2) \in \theta_l\bbZ \times \bbZ, \\ \bu \neq \bzero}}  \bK\left(u_1 \bt_1(\bs) + u_2 \bt_2(\bs) \right).
	\end{equation*}
	Here, $\calS$ is the parametric domain of the film, $\calR$ is the inverse of the curvature, $\theta_l$ is the angular width of the unit cell, and $\bt_i(\bs)$, $i=1,2$, are the tangent vectors at coordinate $\bs\in \calS$. 
	For simplicity, we fix $\bs\in \calS$ and assume $\bt_1 = \be_1$ and $\bt_2 = \be_2$; then the lattice sum above is over a 2-d lattice in $(\be_1, \be_2)$ plane. 
	By substituting the form of $\bK$ and computing $\bH_0(\bs) \bff(\bs)$, we can show that both the normal and the tangential components of $\bff$ are present in the final expression for the energy above.

 	    We can understand these differences physically, by first noticing that the dimension reduction starting from the 3-d continuum contains minimal information about the detailed geometry of the underlying lattice within the nanostructure; these approaches have 3-d continuum theory as their starting point, and are valid for situations in which the limiting thin object has all dimensions much larger than the atomic lengthscale.
        In contrast, the discrete-to-continuum approach presented here is appropriate for nanostructures in which the thin dimensions are comparable to the atomic lengthscale.
        For this reason, the thin-film model obtained in this work may capture better the electromechanics of lipid bilayers, as these are composed of only 1-2 unit cells in the thickness direction \cite{liu2013flexoelectricity,ahmadpoor2013apparent,torbati2022coupling,steigmann2018mechanics}.

    \section*{Acknowledgments}
    This paper draws from the doctoral dissertation of Prashant K. Jha at Carnegie Mellon University \cite{jha2016coarse}.
    We thank Richard D. James for useful discussions; AFRL for hosting a visit by Kaushik Dayal; and NSF (2108784, 1921857), ARO (W911NF-17-1-0084), ONR (N00014-18-1-2528), and AFOSR (MURI FA9550-18-1-0095) for financial support.
    
\newcommand{\etalchar}[1]{$^{#1}$}

\end{document}